\documentclass[11pt]{article} 
\usepackage[letterpaper, margin=1in]{geometry} 
\usepackage{url}
\usepackage{graphicx} 
\usepackage{amssymb,amsmath,amsthm}
\usepackage{thmtools, thm-restate}
\usepackage{algorithm}
\usepackage{algorithmic}
\usepackage{enumitem} 
\usepackage{cases} 
\usepackage{xcolor}

\definecolor{WildStrawberry}{RGB}{255,67,164}
\definecolor{Dousha}{RGB}{47, 136, 67}
\definecolor{BlueBerry}{RGB}{44, 96, 189}
\newcommand*\samethanks[1][\value{footnote}]{\footnotemark[#1]}

\newtheorem{theorem}{Theorem}
\newtheorem{lemma}{Lemma}
\newtheorem{claim}{Claim}

\newcommand{\be}{\begin{equation}}
\newcommand{\ee}{\end{equation}}
\newcommand{\beq}{\begin{equation*}}
\newcommand{\eeq}{\end{equation*}}
\newcommand{\argmin}{\mathop{\rm argmin}}

\newcommand{\R}{\mathbb{R}}

\newcommand{\eps}{\varepsilon}

\newcommand{\Xcomment}[1]{{}}

\newcommand{\eqdef}{\overset{\mathrm{def}}{=\mathrel{\mkern-3mu}=}}
\newcommand{\vect}[1]{\ensuremath{\mathbf{#1}}}

\newcommand\restr[2]{{
  \left.\kern-\nulldelimiterspace 
  #1 
  \vphantom{\big|} 
  \right|_{#2} 
  }}

\def \reals {{\mathbb R}}
\newcommand{\nats}{{\mathbb N}}

\newcommand{\costi}[1][i]{\texttt{cost}_{#1}}
\newcommand{\SC}{\texttt{SC}}

\newcommand{\qnorm}[2][q]{\left\Vert #2\right\Vert_{#1}}

\newcommand{\lqnorm}[1][q]{\mathbb{L}_{#1}}
\newcommand{\mech}{\mathcal{M}}
\newcommand{\loci}[1][i]{\vect{p}_{#1}}

\newcommand{\loc}{\vect{p}}
\newcommand{\locs}{\vect{P}}
\newcommand{\locsmi}[1][i]{\vect{P}_{\text{-}#1}}
\newcommand{\vl}{\vect{l}}
\newcommand{\vli}[1][i]{\vl_{#1}}
\newcommand{\median}{\vect{m}}
\newcommand{\zerovec}{\vect{0}}
\newcommand{\tmedian}{\texttt{median}}
\newcommand{\facility}{\vect{f}}

\newcommand{\m}{\mathbf{m}}

\newcommand{\energy}[1][S]{\Delta_{{#1}}}
\newcommand{\energyc}[1][S]{{\energy[\overline{#1}]}}
\newcommand{\ecsq}{\Delta_{\overline{S}}^{1/q}}
\newcommand{\esq}{\Delta_{S}^{1/q}}
\newcommand{\energyT}{\energy[T]}
\newcommand{\energyTc}{\energyc[T]}
\newcommand{\ectq}{\Delta_{\overline{T}}^{1/q}}
\newcommand{\etq}{\Delta_{T}^{1/q}}

\newcommand{\flr}[1]{\left\lfloor #1\right\rfloor}

\newcommand{\sign}{\texttt{sign}}
\newcommand{\UB}{\texttt{UB}}
\newcommand{\LB}{\texttt{LB}}

\newcommand{\vx}{\vect{x}}

\newcommand{\CMP}{\texttt{CMP}}
\newcommand{\cmp}{\vect{m_p}}
\newcommand{\pred}{\hat{\mathbf{f}}}

\title{Approximation guarantees of Median Mechanism in $\mathbb{R}^d$}

\author{Nick Gravin\thanks{ITCS, Shanghai University of Finance and Economics} \and Jianhao Jia\samethanks}

\date{}

\begin{document}
\maketitle

\begin{abstract}
The coordinate-wise median is a classic and most well-studied strategy-proof mechanism in social choice and facility location scenarios. Surprisingly, there is no systematic study of its approximation ratio in $d$-dimensional spaces. The best known approximation guarantee in $d$-dimensional Euclidean space $\mathbb{L}_2(\mathbb{R}^d)$ is $\sqrt{d}$ via a simple argument of  embedding $\mathbb{L}_1(\mathbb{R}^d)$ into $\mathbb{L}_2(\mathbb{R}^d)$ metric space, that only appeared in appendix of [Meir 2019].
This upper bound is known to be tight in dimension $d=2$, but there are no known super constant lower bounds. Still, it seems that the community's belief about coordinate-wise median is on the side of $\Theta(\sqrt{d})$. E.g., a few recent papers on mechanism design with predictions [Agrawal, Balkanski, Gkatzelis, Ou, Tan 2022], [Christodoulou, Sgouritsa, Vlachos 2024], and [Barak, Gupta, Talgam-Cohen 2024] directly rely on the $\sqrt{d}$-approximation result.

In this paper, we systematically study approximate efficiency of the coordinate-median in $\mathbb{L}_{q}(\mathbb{R}^d)$ spaces for any $\mathbb{L}_q$ norm with $q\in[1,\infty]$ and any dimension $d$. We derive a series of constant upper bounds $UB(q)$ independent of the dimension $d$. This series $UB(q)$ is growing with parameter $q$, but never exceeds the constant $UB(\infty)= 3$. Our bound $UB(2)=\sqrt{6\sqrt{3}-8}<1.55$ for $\mathbb{L}_2$ norm is only slightly worse than the tight approximation guarantee of $\sqrt{2}>1.41$ in dimension $d=2$. Furthermore, we show that our upper bounds are essentially tight
by giving almost matching lower bounds $LB(q,d)=UB(q)\cdot(1-O(1/d))$ for any dimension $d$ with $LB(q,d)=UB(q)$ when $d\to\infty$. We also extend our analysis to the generalized median mechanism used in [Agrawal, Balkanski, Gkatzelis, Ou, Tan 2022] for $\mathbb{L}_2(\mathbb{R}^2)$ space to arbitrary dimensions $d$ with similar results for both robustness and consistency approximation guarantees.

\end{abstract}

\section{Introduction}
\label{sec:intro}
Strategic facility location is the canonical problem in the literature on mechanism design without money~\cite{procaccia2013approximate}. 
The goal in this problem is to find a good location for the new facility given preferences of $n$ strategic agents. 
This problem has been the focus of a long 
line of work in algorithmic mechanism design literature~\cite{procaccia2013approximate, FotakisT14, FotakisT16, SerafinoV16, walsh2020strategy, GkatzelisKST22}  as well as earlier body of work in social choice~\cite{moulin1980strategy,border1983straightforward,kim1984nonmanipulability,peters1993range,barbera1993generalized,ching1997strategy,peremans1997strategy,barbera1998strategy,schummer2002strategy}. In the former case the setting often served as an important domain for testing new concepts in approximate mechanism design, such as truthful mechanisms without monetary transfers~\cite{procaccia2013approximate} as well as mechanism design with predictions~\cite{GkatzelisKST22,XuL22,barak2024mac}, while in the later case it has been the primary domain with structured preferences that allowed one to escape strong impossibility results such as the Gibbard-Satterthwaite Theorem~\cite{Gibbard77,Satterthwaite75}.

%
The social choice literature very thoroughly studied strategy-proof voting scenarios with single-picked preferences~\cite{moulin1980strategy} in $\reals^1$ and its extensions~\cite{barbera1993generalized,barbera1998strategy,schummer2002strategy,ching1997strategy} usually  to $d$-dimensional Euclidean~\cite{kim1984nonmanipulability,peters1993range,border1983straightforward,peremans1997strategy} spaces $\reals^d$.
All this work crucially relies on the simple \emph{median} selection rule initially proposed by Black~\cite{Black48}. The median rule, which puts facility at the median of the agents' reported peaks, is strategy-proof. Its generalized\footnote{The mechanism simply adds a set of fixed points to the reports of all agents.} variants are the only deterministic strategy-proof mechanisms for agents residing on continuous or discrete lines~\cite{DokowFMN12,schummer2002strategy}, and it is the only strategy-proof mechanism with optimal utilitarian social cost, when agents have single-peaked preferences~\cite{moulin1980strategy}. Furthermore, the median rule naturally extends to strategy-proof \emph{coordinate-wise median} mechanism in higher dimensions $\reals^d$: facility's coordinates are separately computed as medians in each of the $d$ dimensions.
In Euclidean space $\lqnorm[2](\reals^2)$, the coordinate-wise median is the only deterministic strategy-proof mechanism that is Pareto optimal and anonymous~\cite{peters1993range}. 

The computer science community, on the other hand, has been mostly concerned with establishing approximate efficiency guarantees and extending original strategic facility location setting in multiple directions (see~\cite{chan2021mechanism} for a recent survey on the topic). The most notable extensions include randomized mechanisms, which  first appeared in~\cite{AlonFPT10} and are commonly used in general metric spaces~\cite{LuSWZ10,FotakisT10,FeldmanW13}; and 
opening of $k\ge 2$ facilities~\cite{procaccia2013approximate,lu2009,LuSWZ10,FotakisT10,EscoffierGTPS11,FotakisT16,walsh2020strategy}, which among other things showed impossibility of constant approximation deterministic mechanisms for installing $k>2$ facilities in $\reals^1$ and $k=2$ facilities in general metric spaces.
However, despite active research in the last fifteen years since the seminal work of Procaccia and Tennenholtz~\cite{procaccia2013approximate}, we do not have tight efficiency guarantees of the celebrated median mechanism in $\reals^d$. This is the question we address in our work. Before listing our results, let us first do a quick overview of what is known about approximation guarantees of the  coordinate-wise median.

\begin{description}
\item[$d=1$.] The median mechanism is deterministic and achieves $1$-approximation in $\reals^d$ for $d=1$ and utilitarian social cost~\cite{procaccia2013approximate}. 
Feigenbaum et al.~\cite{FeigenbaumSY17} showed that its approximation guarantee in $\lqnorm[p](\reals^1)$ space is $2^{1-1/p}$, i.e., when the social cost is measured as the $\lqnorm[p]$ norm of agents' utilities for $p\in[1,+\infty]$. 
\item[$d=2$.] 
The coordinate-wise median is $\sqrt{2}$-approximation to the optimum in the Euclidean space $\lqnorm[2](\reals^2)$~\cite{meir2019strategyproof}. This guarantee is tight~\cite{GoelHC23} for $d=2$. Goel and Hann-Caruthers~\cite{GoelHC23} also considered other $\lqnorm(\reals^2)$ spaces. They showed that median rule has the lowest worst-case approximation ratio among deterministic, anonymous, and strategyproof mechanisms and that its approximation ratio is between $[2^{1-\frac{1}{q}},2^{\frac{3}{2}-\frac{2}{q}}]$ for $\lqnorm(\reals^2)$ spaces.
\item[$d>2$.]  Meir~\cite{meir2019strategyproof} also showed that the $\sqrt{2}$ bound extends to $\sqrt{d}$-approximation guarantee in $\lqnorm[2](\reals^d)$ spaces for any $d\ge 1$. This result\footnote{While it could have been known before, we are only aware of one source, where it was explicitly given -- the appendix of Meir's paper~\cite{meir2019strategyproof}.} for $\lqnorm[2](\reals^d)$ easily follows from two very simple observations: coordinate-wise median is optimal in $\lqnorm[1](\reals^d)$, and the canonical embedding of $\lqnorm[1](\reals^d)$ space into $\lqnorm[2](\reals^d)$ space has distortion\footnote{We refer here to the standard mathematical term ``distortion of metric embedding'', which is calculated as the ratio between maximal and minimal multiplicative discrepancies between two distance functions.} of $\sqrt{d}$. This bound of $\sqrt{d}$ is the best known up to date (see, e.g.,~\cite{barak2024mac}), and, rather surprisingly, no matching lower bound of order $\sqrt{d}$ are known for higher dimensions.
\item[Other Metrics.] There has been a lot of interest in extending it beyond $\reals^d$ 
to a broader set of metric spaces, e.g., given by shortest path distances on a network. However, for a simple cycle metric, any onto and strategy-proof mechanism must be dictatorial~\cite{schummer2002strategy}, and the worst-case approximation ratio of any deterministic strategy-proof mechanism is of order $\Omega(n)$~\cite{DokowFMN12}. This essentially limits the scope of possible application scenarios to trees~\cite{DokowFMN12,FeldmanW13}, or to the instances with $n=O(1)$ a constant number of agents~\cite{meir2019strategyproof}. Furthermore, the appropriately defined median mechanism is strategy-proof and achieves optimal social cost on trees~\cite{FilimonovM22}.
\end{description}
In fact, a very recent line of work on mechanism design with predictions heavily relies on the $\sqrt{d}$-approximation guarantee of the median rule, either by using Meir's result in a black box manner~\cite{barak2024mac}, or significantly limiting their attention only to $\lqnorm[2](\reals^2)$ space~\cite{AgrawalBGTX22,BalkanskiGS24,ChristodoulouSV24}. It also appears that the consensus in the community about 
approximation ratio of the coordinate-wise median mechanism is on $\sqrt{d}$ side, and that we are simply missing examples with matching lower bounds.

\subsection{Our Results}
\label{sec:results}
Following~\cite{FeigenbaumSY17}, we also investigate other than Euclidean norms $\lqnorm$ in dimensions $d>1$. Thus our main question is
\begin{quote}
What are the approximation ratios of median mechanism in $\lqnorm(\reals^d)$ for any $d\in\nats$ and $q\in\reals_{\ge 1}\cup\{\infty\}$?  
\end{quote}
As it turns out, the common belief of $\sqrt{d}$ approximation is wrong for any value of $q$. Specifically, we obtain the following results:    
\begin{enumerate}
        \item We show that coordinate-wise median has approximation ratio of at most $3$ in $\lqnorm(\reals^d)$ for any values of $d$ and $q$. This improves 
        previous upper bound of $\sqrt{d}$ for the Euclidean distance, and also gives new constant upper bounds for all $\lqnorm$ norms.  
        \item We formulate the question about approximation guarantee as a certain mathematical program and then consider appropriate relaxations. In particular, it is 
        convenient to make $d\to\infty$. After carefully examining various local optimum conditions we reduce the relaxed problem to a simple system of differential equations. 
        For each value of $q\in[1,\infty)$ we find solutions to this system and derive our upper bound $\UB(q)$ on the approximation ratio. 
        \item Our analysis is tight, as we demonstrate with a series of examples that yield lower bounds $\LB(q,d)$ approaching our upper bounds $\UB(q)$ at the rate 
        $O(\frac{1}{d})$ for every $q$. For the most well studied Euclidean $\lqnorm[2]$ norm, the difference between our $\UB(2)=\sqrt{6\sqrt{3}-8}<1.55$ and the existing 
        lower bound of 
        $\sqrt{2}>1.41$ is already quite small (less than $10\%$).
        \item By improving $\sqrt{d}$-approximation to a small constant, our work implies better guarantees in the facility location setting with MAC advice of Barak, Gupta, and Talgam-Cohen~\cite{barak2024mac}. We also apply our optimization framework to the generalized median mechanism of Agrawal et al.~\cite{AgrawalBGTX22} in 
        the setting of facility location with predictions. We derive slightly worse (by a similar factor of $\approx 10\%$ than for $d=2$) approximation guarantees for both consistency and robustness of the mechanism that holds in arbitrary dimension $d$.
    \end{enumerate}

\subsection{Additional Related Work}
\label{sec:related}
%
The literature on strategic facility location is rather vast to overview it here (see two surveys on the topic from perspectives of  approximate mechanism design~\cite{chan2021mechanism} and social choice~\cite{barbera2011strategyproof}). We only mention a few directions explored in the computer science literature: various objectives on social cost like square of the $\lqnorm[2]$~\cite{FeldmanW13}, obnoxious facility~\cite{cheng2013strategy,mei2019}, dual preferences~\cite{FeigenbaumS15,zou15}, concave costs~\cite{FotakisT16}, distributed selection of multiple facility~\cite{filoratsikas2021}, false-name-proof mechanisms~\cite{todo11}. 

In the social choice literature, all characterization results assume certain form of single-picked preferences. They either formalize extensions of single-pick condition in $\reals^2$ from~\cite{moulin1980strategy} via axiomatic properties 
~\cite{barbera1993generalized,barbera1998strategy,schummer2002strategy,ching1997strategy}, or assume that preferences are given by $\lqnorm[2](\reals^d)$~\cite{kim1984nonmanipulability,peters1993range} (directly applies to strategic facility location), or work with the classes of separable quadratic preferences in $\reals^d$ space~\cite{border1983straightforward,peremans1997strategy}. The last type of preferences is described by weighted $\lqnorm[2]$ norm with potentially different weights on coordinates. While related to facility location setting in $\lqnorm[2](\reals^d)$, these classification results do not directly apply to strategic facility location, as the voters may have differently weighted quadratic preferences.

The facility location problem serves as an important case study for the new frameworks of mechanism design with predictions. Other examples, 
of learning-augmented mechanisms without payments 
include scheduling~\cite{XuL22,BalkanskiGT23,ChristodoulouSV24},  allocation of public goods~\cite{BanerjeeGHJM023}, and house allocations~\cite{ChristodoulouSV24}. There also non truthful rules for strategic facility locations with predictions that achieve good price of anarchy guarantees~\cite{ChenGI24}.

\section{Preliminaries}
\label{sec:prelim}
In a (strategic) facility location problem the goal is to place a facility $\facility$ to serve a set of $n$ agents residing in a 
$d$-dimensional vector space $\reals^d$. Each agent $i\in[n]$ reports their position $\loci=(p_{i,1},p_{i,2},\cdots,p_{i,d})\in \reals^d$ to the central planner, who upon receiving all reports $\locs=(\loci)_{i\in[n]}$ runs a mechanism $\mech$ that places the facility  $\facility=\mech(\locs)\in\reals^d$. 

We assume that the vector space $\reals^d$ is equipped with an $\lqnorm$ norm $\qnorm{\cdot}$ defined for each vector $\vx\in\reals^d$ and   $q\in\reals_{\ge 1}\cup\{+\infty\}$ as 
$
\qnorm{\vx}=\left(|x_1|^q+\ldots+|x_d|^q\right)^{1/q}
$ and $\qnorm[\infty]{\vx}=\lim\limits_{q\to\infty}\qnorm{\vx}=\max\limits_{j\in[d]}|x_j|$.  Each agent $i\in [n]$ has the most preferred location $\vli\in\reals^d$ and incurs a cost of $\costi(\facility,\vli)=\qnorm{\facility-\vli}$, when facility is placed at $\facility\in\reals^d$. Each agent $i$ may strategically misreport their location $\loci\ne\vli$ to minimize their own cost $\costi(\facility(\loci,\locsmi),\vli)$. A mechanism $\mech(\locs)$ is called \emph{strategy-proof} (also  \emph{truthful}) if truth telling $\loci=\vli$ is the dominant strategy for each $i\in[n]$ on any reports of other agents $\locsmi=(\loci[1],\ldots,\loci[i-1],\loci[i+1],\ldots,\loci[n])$.  

Thus, for a truthful mechanism $\mech$ one can assume that agent reports $\locs=(\vli)_{i\in[n]}$ are the same as their true locations. The social planner's objective is to minimize the social cost given by $\SC(\locs,\facility)\eqdef\sum_{i\in[n]}\costi(\facility,\loci)=\sum_{i\in[n]}\qnorm{\facility-\loci}$ for a facility placed at $\facility\in\reals^d$.
The performance of $\mech$ is measured as the approximation ratio of the mechanism's social cost $\SC(\locs,\mech(\locs))$ compared to the optimum. We denote the optimal placement of the facility as $\facility\eqdef\argmin\limits_{\vl\in\reals^d}\SC(\locs,\vl)\in\reals^d$. Formally, the \emph{approximation ratio} $\alpha(\mech)$ of the mechanism $\mech$ is defined as the worst-case ratio
\[
\alpha(\mech) \eqdef
\max\limits_{n,\locs}
\frac{\SC(\locs,\mech(\locs))}{\SC(\locs,\facility)}.
\]

The \emph{coordinate-wise median} mechanism $\median(\locs)$ (also 
called median-point or median mechanism) is a well-known 
truthful mechanism. It calculates a median point 
$\median=(m_1,\ldots,m_d)$ with\footnote{For even number of points $n$ the tie in the selection of the median can be broken arbitrarily (for the smaller or larger value) but consistently for each coordinate $j\in[d]$.}
\[
\forall j\in[d]\quad
m_j\eqdef\tmedian\left(p_{1,j},p_{2,j},\ldots,p_{n,j}\right).
\]
The median mechanism outputs the optimal location $\median=\facility$, if the dimension $d=1$, or if there are only $n\le 2$ agents, or in $\lqnorm[1](\reals^d)$ spaces for any $d$ and $n$.  

\section{Approximation Guarantees}
\label{sec:approx}
We first derive a family of upper bounds on the approximation ratio of the coordinate-wise median mechanism in $\lqnorm$ for each value of $q\in[1,\infty)$ (note that the approximation guarantee for $\lqnorm[\infty]$ can be easily obtained by taking the limit of approximation results for $q\to\infty$). We then conclude this section with examples that demonstrate tightness of our upper bound guarantees for each value of $q\in\reals_{\ge 1}\cup\{\infty\}$
and unlimited dimension $d$. Specifically, these examples provide a family of lower bounds $\LB(q,d)$ for each value of $q$ and $d$, which converge to our upper bound $\UB(q)$ at a rate $\LB(q,d)\ge\UB(q)-O(\frac{1}{d})$.  

To derive the upper bounds $\UB(q)$ on the approximation ratio, we write a natural optimization problem over the set of $n$ points $\locs=(\loci)_{i=1}^n$ for a fixed positions of the facility $\facility\in\reals^d$ and the constraint that $\locs$ have a given coordinate-wise median $\median\in\reals^d$. Specifically, we encode the approximation ratio as a linear objective $\SC(\locs,\facility)-\lambda\cdot\SC(\locs,\median)$ for a constant 
$\lambda\in(0,1)$ that captures the inverse of the approximation 
ratio $1/\alpha$ of the coordinate-wise median mechanism $\median(\locs)$.
Our goal is to find $\lambda(q)>0$ such that the minimum of this optimization problem is at least $0$.

We normalize the parameters $(\facility,\median)$ of our optimization problem by applying simple translation, scaling, and reflection transformations of space $\reals^d$ so that: (i) $\median=\vect{0}=(0,\ldots,0)$; (ii) optimal facility $\facility=(f_1,\ldots,f_d)$ has all positive coordinates $f_j> 0$ (if $f_j=0$ we simply let $f_j=\eps$ for arbitrary small $\eps>0$ with only a negligible increase to the value of our optimization formulation); (iii) $\qnorm{\facility}=1$. We also assume that the total number of points $n$ is even, as we can duplicate all points in our instance without changing the locations of $\facility$ and $\median$. 
Our objective function $\SC(\locs,\facility)-\lambda\cdot\SC(\locs,\median)$ can be conveniently separated into $\sum_{i\in[n]}g(\loci)$, where $
g(\loci)\eqdef \qnorm{\loci-\facility}-\lambda\cdot\qnorm{\loci}$. To capture the constraint that $\median$ is a median of the set $\locs$, we define the \emph{signature} $\sigma(\loc)\in\{-1,1\}^d$ of a point $\loc=(p_1,\ldots,p_d)\in\reals^d$, where $\sigma(\loc)\eqdef(\sign(p_1),\ldots,\sign(p_d))$.    
The sign function $\sign(p_{j})=1$ if $p_{j}>0$ and $\sign(p_{j})=-1$ if $p_{j}<0$; when $p_{j}=0$, the sign can take either value $1$ or $-1$ (in the former case we write $p_{j}=0^+$ , and in the latter case we write $p_{j}=0^-$). Hence, our optimization problem can be succinctly written as follows.
\begin{align}
 \label{eq:main}
 \min\limits_{\locs}& \sum_{i\in[n]} g(\loci), &\quad 
 where~ g(\loci)\eqdef\qnorm{\loci-\facility}-\lambda\cdot\qnorm{\loci} \nonumber\\
 \text{s.t.}&  \sum_{i\in [n]}\sigma(\loci) = \vect{0}, & 
\end{align}
We first observe that \eqref{eq:main} indeed achieves minimum (not an infinum) at a certain configuration of points $\locs=(\loci[1],\ldots,\loci[n])$ for each fixed value of $\lambda=\lambda(q)\in(0,1)$.
\begin{claim}
    \label{cl:compact set}
    The optimization problem~\eqref{eq:main} attains minimum.
\end{claim}
\begin{proof}
Consider a closed ball $B_R(\zerovec)\subset{\R^d}$ centered at $\zerovec$ with radius $R=\frac{2n}{1-\lambda}$. The continuous function $\sum_{i\in[n]}g(\loci)$ attains minimum on a compact set $\left\{\locs\in B_R(\zerovec)^n ~\mid~\sum_{i\in[n]}\sigma(\loci)=\vect{0}\right\}$.
Note that by triangle inequality $g(\loc)\ge\qnorm{\loc}-
\qnorm{\facility}-\lambda\cdot\qnorm{\loc}=(1-
\lambda)\cdot\qnorm{\loc}-\qnorm{\facility}\ge -1$ for any 
$\loc\in\reals^d$. By setting $\locs=
(\vect{0},\ldots,\vect{0})$ we get the value of 
$n\cdot\qnorm{\facility}=n$ in \eqref{eq:main}. Now, if any of the points $\loc\in\locs$ is outside of $B_R(\zerovec)$, then $g(\loc)> 2n-1$ and $\sum_{i\in[n]}g(\loci)> 2n-1 + (-1)\cdot (n-1)=n$, i.e., any $\locs$ outside of our compact set achieves larger value than $\locs=(\zerovec,\ldots,\zerovec)$. 
\end{proof}

To solve~\eqref{eq:main}, we first explore the local properties of the optimal solution. Namely, we fix signatures of all locations $\loc\in\locs$ and focus on a local optimization of $g(\loc)$ for a single location $\loc$ with a fixed signature $\sigma(\loc)=\sigma^o\in\{-1,1\}^d$. 
I.e., we find $\loc=\argmin\limits_{\vl:\sigma(\vl)=\sigma^o} g(\vl)$. 

\subsection{Optimization on individual points}
\label{sec:upper_local}
For a fixed signature $\sigma^o$ of $\loc$ we define 
$S(\loc)$ and $\overline{S}(\loc)$ as
$$
S\eqdef\{j\in [d]\ \vert \sigma_{j}^{o}=1\}
\quad\quad
\overline{S}\eqdef\{j\in [d]\ \vert \sigma_j^o=-1\}
$$
the sets of respectively positive and negative coordinates in $\loc$. Therefore, we need to minimize $g(\loc)$ for a $\loc\in\reals^d$ with  $p_j\in[0^+,\infty)$ when $j\in S=S(\loc)$ and $p_j\in(-\infty,0^-]$ when $j\in\overline{S}=\overline{S}(\loc)$. 

We begin by making a simple observation that the optimal $\loc$ may not have $p_j\notin[0^+,f_j)$ for each coordinate $j$. Namely,
\begin{claim}
    \label{cl:coordinate-wise condition}
    We have $p_j\in[f_j,+\infty)$ for each $j\in S$.
\end{claim}
\begin{proof}
If there exists a coordinate $j$ with $p_j\in[0^+,f_j)$, we can change $p_j$ to $f_j$. This does not affect the signature $\sigma(\loc)$. At the same time, it decreases $|p_j-f_j|$ making $\qnorm{\loc-\facility}$ smaller, and it increases $p_j$ making the value of $\qnorm{\loc}$ larger. I.e., $g(\loc)$ decreases, which means that $\loc$ was not optimal.
\end{proof}

We next explore local optimum conditions of $g(\loc)$, which allow us to explicitly find the values of each coordinate $p_j$, when $S\ne\emptyset$ and also narrow down the set of candidate locations for the case of $S=\emptyset$. The explicit form of optimal $\loc$ with given sets of positive and negative coordinates is summarized in the next Lemma~\ref{lm:local optimal condition}. We first define a non negative number $c(\loc)$ as  
    \be
    \label{eq:constant_c}
    c\eqdef
    \lambda^\frac{1}{q-1}\cdot\frac{\qnorm{\loc-\facility}}{\qnorm{\loc}} 
    \ee  
\begin{lemma}[Optimal $\loc$ with given signature $\sigma^o$]
    \label{lm:local optimal condition}
    If
     \begin{enumerate}[label=(\roman*)]
        \item $S(\loc)\ne\emptyset$. Then $\frac{p_j-f_j}{p_j} = c$ for $j\in S$,  $p_j=0^-$ for $j\in \overline{S}$;
        \item$S(\loc)=\emptyset$. Then $\exists\ T\subseteq [d]$ such that $\frac{-p_j+f_j}{-p_j}=c$ for $j\in T$, $p_j=0^-$ for $j\notin T$.
    \end{enumerate}
\end{lemma}
\begin{proof}
We distinguish the optimum solution $\loc$ with the variable $\vl$ in the optimization problem for $g(\vl)$. It is important to note that the function $g(\vl)=\qnorm{\vl-\facility}-\lambda\cdot\qnorm{\vl}$ is differentiable in $l_j$ when $l_j\notin \{0,f_j\}$. Hence, for every coordinate $p_j\notin\{0,f_j\}$ of the optimum  we have $
\frac{\partial}{\partial l_j}g(\vl)\big\vert_{\vl=\loc}=0
$. I.e., 
\begin{equation*}
    \left\{
    \begin{aligned}
    &\frac{\partial}{\partial l_j}g(\vl)\big\vert_{\vl=\loc}=\frac{(p_j-f_j)^{q-1}}{\qnorm{\loc-\facility}^{q-1}}-\lambda\cdot\frac{p_j^{q-1}}{\qnorm{\loc}^{q-1}}=0,\ &p_j>f_j\\
    &\frac{\partial}{\partial l_j}g(\vl)\big\vert_{\vl=\loc}=-\frac{(-p_j+f_j)^{q-1}}{\qnorm{\loc-\facility}^{q-1}}+\lambda\cdot\frac{(-p_j)^{q-1}}{\qnorm{\loc}^{q-1}}=0,\  &p_j< 0
    \end{aligned}
    \right.
\end{equation*}
Thus, for $c$ defined as in~\eqref{eq:constant_c} we get 
\begin{equation}
\label{eq:equation_for_c}
    \left\{
    \begin{aligned}
    &c = \frac{p_j-f_j}{p_j},\ &\text{when }p_j\in (f_j,+\infty)\\
    &c = \frac{f_j-p_j}{-p_j},\  &\text{when }p_j\in (-\infty,0)
    \end{aligned}
    \right.
\end{equation}
Furthermore, if $p_j=f_j$, the function $g(\vl)$ has the right-hand partial derivative $\frac{\partial}{\partial_{+}l_j}g(\vl)$ at $\vl=\loc$. As $\loc$ is the local minimum, $\frac{\partial}{\partial_+ l_j}g(\vl)\big\vert_{\vl=\loc}=\frac{(p_j-f_j)^{q-1}}{\qnorm{\loc-\facility}^{q-1}}-\lambda\cdot\frac{p_j^{q-1}}{\qnorm{\loc}^{q-1}}\ge 0$ when $p_j=f_j$, i.e., 
$
\frac{p_j-f_j}{p_j}\ge
\lambda^{\frac{1}{q-1}}\cdot\frac{\qnorm{\loc-\facility}}{\qnorm{\loc}}
$. As $c=\frac{\qnorm{\loc-\facility}}{\qnorm{\loc}}\ge 0$ and $p_j=f_j$, we get that $0=\frac{p_j-f_j}{p_j}\ge c\ge 0$. This means that $c=0$ and that \eqref{eq:equation_for_c} holds for a larger range of $p_j\in[f_j,+\infty)$.

Next, observe that $c = \frac{p_j-f_j}{p_j}<1$ in \eqref{eq:equation_for_c} when $p_j\in [f_j,\infty)$ (as $f_j>0$), and that $c=\frac{-p_j+f_j}{-p_j}>1$ when $p_j\in (-\infty, 0)$ for the same 
constant $c=\lambda^\frac{1}{q-1}\cdot\frac{\qnorm{\loc-\facility}}{\qnorm{\loc}}$. Hence, $\loc$ cannot 
simultaneously have coordinates $p_j\in[f_j,\infty)$ and $p_i\in(-\infty,0)$ for $i,j\in[d]$. I.e., when $S(\loc)\ne\emptyset$, we have $c<1$, which means that $p_j=0^-$ for any $j\in\overline{S}$ and $c=\frac{p_j-f_j}{p_j}$ for any $j\in S$. This concludes the proof for the 
case (i). On the other hand, when $S(\loc)=\emptyset$, then equation \eqref{eq:equation_for_c} gives us (ii) for $T\eqdef\{j\in [d] ~|~p_j < 0\}$.
\end{proof}

Lemma~\ref{lm:local optimal condition}, while describing the optimal point $\loc\in\reals^d$ with given signature $\sigma(\loc)$ (equivalently, with given sets $S(\loc)$ and $\overline{S}(\loc)$), still does not explain the relation between this signature and the respective value of $g(\loc)$ in the global optimization problem~\eqref{eq:main}. We express $g(\loc)$ for the optimal point $\loc$ as the function of $\energy[S]$ (and $\energyc[S]$) given by\footnote{We generally use the notation $\energy[X]\eqdef\sum_{j\in X}{f_j}^q$ for any set $X\subseteq[d]$.} 
\[
\energy[S]\eqdef\sum_{j\in S}{f_j}^q\quad\quad
\energyc[S] \eqdef\sum_{j\in \overline{S}}f_j^q=1-\energy[S]
\]
in the next Lemma~\ref{lm:expression for local optimum}, which is essential for handling the main optimization problem~\eqref{eq:main}.  
\begin{lemma}[Expression of $g(\loc)$ as a function of $\Delta_{_S}$]
    \label{lm:expression for local optimum}
    For an optimal $\loc$ with 
    given sets of positive $S$ and negative $\overline{S}$ coordinates, if 
    \begin{enumerate}[label=(\roman*)]
        \item $S(\loc)\ne\emptyset$, then 
        $g(\loc)=\Big(1-\lambda^{q/(q-1)}\Big)^{(q-1)/q}\cdot\ecsq-\lambda\cdot\esq$.
        \item $S(\loc)=\emptyset$, then 
        $g(\loc)\geq\Big(1-\lambda^{q/(q-1)}\Big)^{(q-1)/q}$.
    \end{enumerate}
\end{lemma}
\begin{proof}[Proof sketch]
The proof for the most part is elementary algebraic manipulations. We only present key steps and defer detailed derivations to Appendix~\ref{sec:appendix}.
When $S(\loc)\neq\emptyset$, according to Lemma~\ref{lm:local optimal condition}, $p_j=f_j/(1-c)$ for $j\in S$. Then $g(\loc)=\qnorm{\loci-\facility}-\lambda\cdot\qnorm{\loci}$ is
\begin{align}
\label{eq:g(p)}
g(\loc)=\left(\left(\frac{c}{1-c}\right)^q\cdot\esq+\ecsq\right)^{1/q}-\lambda\cdot\frac{1}{1-c}\cdot\ecsq.
\end{align}
Using that $c=\lambda^{1/(q-1)}\cdot\qnorm{\loc-\facility}/\qnorm{\loc}$, we obtain the following equations that allow us to replace terms dependent on $c$ in \eqref{eq:g(p)}:
\begin{align*}
&\frac{c}{1-c}=\left(\frac{\energyc}{\energy}\right)^{1/q}\cdot\left(\frac{\lambda}{1-\lambda^{q/(q-1)}}\right)^{1/q},\quad\text{and}\quad\frac{1}{1-c}=\frac{c}{1-c}+1.
\end{align*}

\[
\text{Then,}\quad g(\loc)=\left(\left(\frac{c}{1-c}\right)^q\cdot\energy+\energyc\right)^{1/q}-\lambda\cdot\frac{1}{1-c}\cdot\ecsq
=\left(1-\lambda^{q/(q-1)}\right)^{(q-1)/q}\cdot\ecsq-\lambda\cdot\esq.
\]
This concludes the proof for the case $S(\loc)\ne\emptyset$. When $S(\loc)=\emptyset$ and $T=\{j~|~p_j< 0^-\}\neq\emptyset$, we obtain similar expression of $g(\loc)$ for a fixed set $T$: 
\begin{align*}
g(\loc)=
        \left(1-\lambda^{q/(q-1)}\right)^{(q-1)/q}\cdot(1-\energyT)^{1/q}+\lambda\cdot\etq.
\end{align*}
We want to understand what is the minimum of the above expression for $g(\loc)$ over 
all possible $T\subseteq[d]$, $T\ne\emptyset$ (when $T=\emptyset$, $g(\loc)=1$). Let 
$x\eqdef\energyT$, then $x\in[0,1]$ and $g(\loc)=t(x)\eqdef\left(1-\lambda^{q/(q-1)}\right)^{(q-1)/q}\cdot(1-x)^{1/q}+\lambda\cdot x^{1/q}$. Note that when 
$T\ne\emptyset$, we have $c>1$ in Lemma~\ref{lm:local optimal condition}, which 
implies that $\energyT\le\lambda^{q/(q-1)}$. Furthermore, as $t(x)$ is a positive 
linear combination of concave functions $x^{1/q}$ and $(1-x)^{1/q}$, the function 
$t(x)$ is concave on $x\in[0,\lambda^{q/(q-1)}]$ and, thus, achieves its minimum at 
the end points $x=0$, or $x=\lambda^{q/(q-1)}$. I.e., $t(0)=\left(1-\lambda^{q/(q-1)}\right)^{(q-1)/q}< 1=t(\lambda^{q/(q-1)})$ is the minimum value. Interestingly, we 
cannot achieve $x=\energyT=0$ for a nonempty $T$ and, on the other hand, when $T=\emptyset$ the point $\loc=\zerovec$ is not a local optimum (in this case $g(\zerovec)=1>\left(1-\lambda^{q/(q-1)}\right)^{(q-1)/q}$). We still get the desired lower bound on $g(\loc)\ge\left(1-\lambda^{q/(q-1)}\right)^{(q-1)/q}$ for the case $S(\loc)=\emptyset$.
\end{proof}

\subsection{Upper Bounds on Approximation Ratios}
\label{sec:upper_global}
We finally derive the family of upper bounds $\UB(q)$ on the approximation ratios $\alpha(q)$ of the median mechanism in $\lqnorm(\reals^{d})$.
\begin{theorem}
    \label{thm:UB}
    The median mechanism in the normed vector space $\lqnorm(\reals^{d})$ with $q\in\reals_{\ge 1}$ of arbitrary dimension $d\in\nats_{\ge 1}$ has
    approximation ratio $\alpha(q)\le\UB(q)$, where $\UB(q)$ is an increasing function in $q$ with $\UB(1)=1$ and $\UB(q)\to 3$ when $q\to\infty$.
    For $q=2$, $\UB(2)=\sqrt{6\sqrt{3}-8}\approx 1.55$.
\end{theorem}
\begin{proof} We now consider our main optimization problem~\eqref{eq:main}.
Lemma~\ref{lm:expression for local optimum} shows that, for any location $\loci\in\locs$, $g(\loci)$ can be expressed\footnote{In case $\energy[S(\loci)]=0$, we only have a lower bound on $g(\loci)$, which is without any loss of generality when dimension $d\to\infty$.} as a function of $\energy[S(\loci)]$ and $\energy[\overline{S}(\loci)]=1-\energy[S(\loci)]$. 
The constraints $\sum_{i\in [n]}\sigma(\loci) = \vect{0}$ of~\eqref{eq:main} is still unwieldy. On the other hand, it implies that
\begin{equation}
    \label{eq:relaxed_median_constraint}
    \sum\limits_{i\in[n]}\energy[S(\loci)]=
    \sum_{i\in[n]}\sum_{j:p_{i,j}\ge 0^+}f_j^q
    =\sum_{j\in[d]}f_j^q\cdot\sum_{i:p_{i,j}\ge 0^+}1
    =\frac{n}{2}\cdot\sum_{j\in[d]}f_j^q=\frac{n}{2}.
\end{equation}
We simply relax the constraint in \eqref{eq:main} to \eqref{eq:relaxed_median_constraint} and solve the corresponding minimization problem (recall that we are interested in finding the largest possible $\lambda<1$ such that \eqref{eq:main} is at least $0$). To further simplify the presentation of the relaxed optimization problem, we let $x_i\eqdef\energy[S(\loci)]\in[0,1]$, $1-x_i=\energy[\overline{S}(\loci)]$, and define $\delta=\delta(\lambda)\eqdef\left(1-\lambda^{q/(q-1)}\right)^{(q-1)/q}/\lambda=\left(\lambda^{-q/(q-1)}-1\right)^{(q-1)/q}$, which is a decreasing function of $\lambda$.
Then the objective function of \eqref{eq:main} can be rewritten as $\sum_{i\in[n]}h(x_i)$, where $h(x_i)\eqdef \lambda\cdot\left(\delta\cdot(1-x_i)^{1/q}-x_i^{1/q}\right)$. The relaxed problem in $\vx=(x_i)_{i\in[n]}$ is as follows
\begin{align}
 \label{eq:relaxed main}
 \min\limits_{\vx}& \sum_{i\in[n]} h(x_i), &&\quad
 where~ h(x_i)\eqdef \lambda\cdot\left(\delta\cdot(1-x_i)^{1/q}- x_i^{1/q}\right)\nonumber\\
 \text{s.t.}&  \sum_{i\in [n]}x_i = \frac{n}{2},
 &&\quad\quad \forall i\in[n]~x_i\in[0,1]
\end{align}
Note that we are only interested in $\lambda$ such that $\delta(\lambda)\ge 1$ (otherwise, \eqref{eq:relaxed main} is negative for $(x_i=1/2)_{i\in[n]}$). Next, to solve \eqref{eq:relaxed main}, we explore convexity/concavity of $h(\cdot)$ on the interval $x\in[0,1]$.
\begin{claim}
    \label{cl:property of h(x)}
    The function $h(x)$ is convex in $x$ on $[0,z]$ and concave in $x$ on $[z,1]$, where $z\eqdef\delta^{-q/(2q-1)}/\left(\delta^{-q/(2q-1)}+1\right)\le 1/2$.
\end{claim}
\begin{proof}
We take second-order derivative of $h(x)$:
\begin{align*}
    \frac{\partial^2}{\partial x^2}h(x)=\lambda\cdot\frac{q-1}{q^2}\cdot\left(x^{-(2q-1)/q}-\delta\cdot(1-x)^{-(2q-1)/q}\right).
\end{align*}
The function $\frac{\partial^2}{\partial x^2}h(x)$ goes from $+\infty$ to $-\infty$ when $x$ goes from $0$ to $1$. The only real root of $\frac{\partial^2}{\partial x^2}h(x)$ on $x\in[0,1]$ is 
\[
z=\delta^{-q/(2q-1)}/\left(\delta^{-q/(2q-1)}+1\right).
\] 
As $\delta\ge 1$, we have $z\le\frac{1}{2}$.
Thus, $\frac{\partial^2}{\partial x^2}h(x)\ge 0$
for $x\in[0,z]$ implying that $h(x)$ is convex on $[0,z]$, and $\frac{\partial^2}{\partial x^2}h(x)\le 0$
for $x\in[z,1]$ implying that $h(x)$ is concave on $[z,1]$.
\end{proof}
Given Claim~\ref{cl:property of h(x)}, we can succinctly describe local minima of~\eqref{eq:relaxed main}.
\begin{lemma}[Possible Solutions to~\eqref{eq:relaxed main}]
    \label{optimal value of x_i}
    $\exists a\in[0,z],b\in[z,1)$, such that
    $\forall i\in[n]~x_i\in\{a, b, 1\}$ in the optimal solution to \eqref{eq:relaxed main}. Moreover, $\vert\{i\in[n]~\vert~ x_i=b\}\vert\le 1$. 
\end{lemma}
\begin{proof}
First, we show that $\exists a\in[0,z]$ such that each $x_i=a$ if $x_i\in[0,z]$. Assume towards a contradiction that $0\le x_1 < x_2\le z$ in \eqref{eq:relaxed main}. As $h(\cdot)$ is convex on $[0,z]$, we could replace $x_1$ and $x_2$ with $x_1+\eps$ and $x_2-\eps$ without violating the constraint $\sum_{i\in[n]}x_i=n/2$ and reduce $h(x_1)+h(x_2)>h(x_1+\eps)+h(x_2-\eps)$.

Next, we show that there could be at most one $x_i\in(z,1)$ in the optimum solution to \eqref{eq:relaxed main}. We assume that $z<x_1\leq x_2<1$ towards a contradiction. Then we could decrease $h(x_1)+h(x_2)>h(x_1-\eps)+h(x_2+\eps)$ by moving $x_1$ and $x_2$ to $x_1-\eps$ and $x_2+\eps$ without violating the constraint $\sum_{i\in[n]}x_i=n/2$.
\end{proof}
Lemma~\ref{optimal value of x_i} allow us to reduce the space of optimization~\eqref{eq:relaxed main} to only three parameters $a,b,$ and the number of points $|\{i: x_i=a\}|$. Moreover, since we may duplicate the locations of optimal solution $\locs$ without changing the objectives and violating the constraints in~\eqref{eq:main} and~\eqref{eq:relaxed main}, we can assume that the minimum in~\eqref{eq:main} and~\eqref{eq:relaxed main} is attained when the number of points $n$ goes to infinity. Now, when $n\to+\infty$ we can ignore the effect of a single point\footnote{Formally, for any $\lambda\in(0,1)$ such that~\eqref{eq:relaxed main} is strictly smaller than $0$ for $n=n_0$ points, we let $n=100\cdot n_0$ and construct a feasible solution $\vx$ to \eqref{eq:relaxed main}, such that all $x_i\in\{a,1\}$ and $\sum_{i\in[n]}h(x_i)<0$.} $x_i=b$ in~\eqref{eq:relaxed main}, which reduces the number of parameter down to $2$ ($a$, $|\{i: x_i=a\}|$). As we have the constraint $\sum_{i=1}^n x_i = n/2$, the number $|\{i: x_i=a\}|$ and $|\{i: x_i=1\}|$ must be respectively $\frac{n}{2(1-a)}$ and $\frac{1-2a}{2(1-a)}\cdot n$. 
I.e., there is only one parameter $a$ to optimize in~\eqref{eq:relaxed main}. Recall that we need to find $\lambda$ (which defines $\delta(\lambda)$) such that   $\sum_{i\in[n]}h(x_i)\ge 0$ for any feasible $\vx$
in~\eqref{eq:relaxed main}.
Equivalently, we need to find $\lambda$ such that
\begin{align}
    \label{eq: optimization over a}   &\min_{a\in[0,z]}\left[\frac{n}{2(1-a)}\cdot h(a) + \frac{1-2a}{2(1-a)}\cdot n\cdot h(1)\right]&\\
    &=\min_{a\in[0,z]}\left[\frac{n}{2(1-a)}\cdot\lambda\cdot\left(\delta\cdot(1-a)^{1/q}-a^{1/q}-1+2\cdot a\right)\right]\ge 0,\nonumber
\end{align}
Thus, we need to minimize $u(a)\eqdef\delta\cdot(1-a)^{1/q}-a^{1/q}-1+2\cdot a$.
Let $a^*\eqdef \argmin\limits_{a\in[0,z]} u(a)$. We observe that 
\begin{lemma}
    \label{lm:optimal solution for a}
    If $u(a^*)=0$, then $a^*$ is optimal $\Rightarrow u'(a^*)=0$.
\end{lemma}
The proof of Lemma~\ref{lm:optimal solution for a} is straightforward (we simply verify that if $u(a^*)=0$, $a\in\{0,z\}$ are not optimal, and that $u'(a)=0$ has a unique solution on $(0,z)$) and we defer it to Appendix. 

The final step in the proof of Theorem~\ref{thm:UB} is to find $\lambda^*$ such that the minimum of \eqref{eq:relaxed main} and respectively \eqref{eq: optimization over a} is equal to $0$.
As \eqref{eq:relaxed main} is a relaxation of \eqref{eq:main} with a smaller or equal value, we immediately would get that \eqref{eq:main} has value at least $0$ for $\lambda=\lambda^*$, which translates into $1/\lambda^*$ approximation guarantee for the coordinate-wise median mechanism. By Lemma~\ref{lm:optimal solution for a}, the optimal $\lambda^*$ and $a^*$ must satisfy equations $u(a^*)=0$ and $u'(a^*)=0$. Then by denoting $\delta^*=\delta(\lambda^*)$, we get the following system of equations on $\lambda^*$ and $a^*$
\begin{equation}
\label{eq:upper bound}
\left\{
    \begin{aligned}
    &\delta^*\cdot(1-a^*)^{\frac{1}{q}}-{(a^*)}^{\frac{1}{q}}-1+2\cdot a^*=0\\
    &\frac{1}{q}\cdot\left(-\delta^*\cdot(1-a^*)^{\frac{1-q}{q}}-(a^*)^{\frac{1-q}{q}}\right)+2=0.
    \end{aligned}
    \right.
\end{equation}
After multiplying by $1-a^*$ the second equation in \eqref{eq:upper bound}, we get
\begin{equation}
\label{eq:inequality of a}
    \frac{1}{q}\cdot\left(-\delta^*\cdot(1-a^*)^{\frac{1}{q}}+(a^*)^{\frac{1}{q}}-(a^*)^{\frac{1-q}{q}}\right)+2-2a^*=0.
\end{equation}
The first equation in \eqref{eq:upper bound} gives us $\delta^*\cdot(1-a^*)^{1/q}-{(a^*)}^{1/q}=1-2\cdot a^*$. By plugging it in \eqref{eq:inequality of a} we get
\begin{align}
    \label{eq:optimal a^*}
    2\cdot\left(1-\frac{1}{q}\right)\cdot a^*+\frac{1}{q}\cdot (a^*)^{\frac{1-q}{q}}-2+\frac{1}{q}=0.
\end{align}
We now obtain $\lambda^*(q)$ through the following steps:
\begin{enumerate}[label=(\roman*)]
        \item find $a^*$ from equation~\eqref{eq:optimal a^*}.
        \item set $\delta^*=\left((a^*)^{1/q}+1-2\cdot a^*\right)/(1-a^*)^{1/q}$ from the first line of~\eqref{eq:upper bound}.
        \item calculate $\lambda^*$ by inverting function $\delta(\lambda)=\left(\lambda^{-\frac{q}{q-1}}-1\right)^{\frac{q-1}{q}}$ at $\delta=\delta^*$:
        \[
        \lambda^*=\left(1+(\delta^*)^{\frac{q}{q-1}}\right)^{-\frac{q-1}{q}}.
        \]
\end{enumerate}
We can numerically calculate $a^*(q)$ for each $q\ge 1$ and then get $\UB(q)=1/\lambda^*(q)$ (see figure~\ref{fig:UB(q)}).
\begin{figure}[htb] 
		\centering
		\includegraphics[width=5.0in]{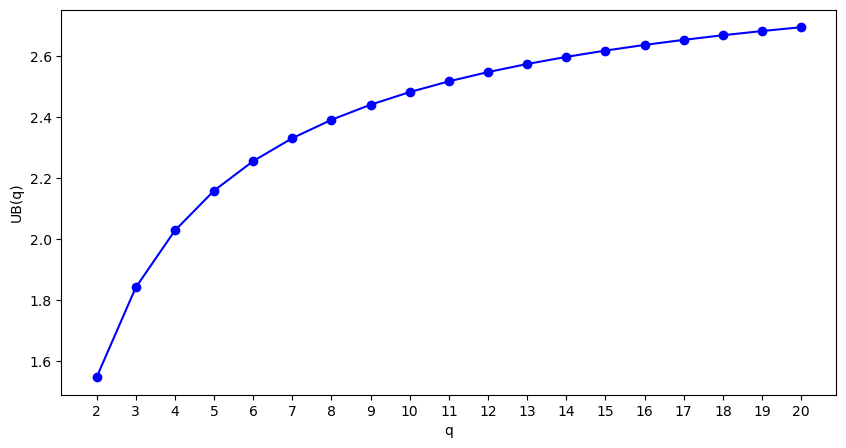}
		\caption{UB(q)}
\label{fig:UB(q)}		
\end{figure}

Furthermore, for $q=2$ and $q\to+\infty$ we can find algebraic expressions for $a^*$ and thus get precise values of $\UB(q)$. In particular, for $q=2$ the equation \eqref{eq:optimal a^*} becomes $a^*+\frac{1}{2}\cdot (a^*)^ {-1/2}-\frac{3}{2}=0$, which gives $a^*=1-\frac{\sqrt{3}}{2}$. By following steps (ii) and (iii) we get $\UB(2)=1/\lambda^*=\sqrt{6\sqrt{3}-8}$.
When $q\rightarrow\infty$, in step (i) one can get $a=\frac{1+o(1)}{2q}$. Then by (ii) we get $\delta^*\rightarrow (a^*)^{1/q}+1\rightarrow 2$ and by (iii) $\lambda^*\rightarrow 1/(1+\delta^*)\rightarrow 1/3$. I.e., we get that median mechanism is a $3$-approximation under $\lqnorm[\infty]$ norm.
\end{proof}

\subsection{Lower Bounds on Approximation Ratios}
\label{sec:lower_global}
In this section, we show that the upper bound $\UB(q)$ is essentially tight by constructing a series of instances for each norm $\lqnorm$ with $q\ge 1$ and each dimension $d\in\nats_{\ge 1}$ such that coordinate-wise median achieves $\LB(q,d)=\UB(q)\cdot\left(1-O(\frac{1}{d})\right)$ approximation. For a given $q$ and $d$ our instance is as follows.

The optimal facility location is at $\facility=(1,1,\cdots,1)\in\reals^d$ and the median $\m=(0,0,\cdots,0)\in\reals^d$. We shall use parameters $a^*,\lambda^*$ given by \eqref{eq:optimal a^*}. We also define parameter $c^*$ (analog of $c$ from Section~\ref{sec:upper_local}) as
\begin{align*}
        c^*\eqdef\frac{\left(\frac{1-a^*}{a^*}\cdot\frac{(\lambda^*)^{q/(q-1)}}{1-(\lambda^*)^{q/(q-1)}}\right)^{1/q}}{1+\left(\frac{1-a^*}{a^*}\cdot\frac{(\lambda^*)^{q/(q-1)}}{1-(\lambda^*)^{q/(q-1)}}\right)^{1/q}}.
\end{align*}        
The number of points $n$ is a very large integer. There are only two types of points $\loc\in\reals^d$ in $\locs$.
\begin{description}
        \item[Type I.] $\flr{a^*\cdot d}$ coordinates are $1/(1-c^*)$ and the rest are $0$.        
        \item[Type II.] $\loc=(1,1,\cdots,1)$ the same as $\facility$.
\end{description}
The numbers\footnote{Obviously, both of these numbers must be integers, but as $n$ can be arbitrary large, we will simply ignore the dependency on $n$ and the fact that $\frac{n}{2-2a^*}$ may be non integer.} of Type I and II points are respectively $\frac{n}{2-2a^*}$ and $n\cdot\frac{1-2a^*}{2-2a^*}$. 
The strictly positive coordinates of Type I points are uniformly distributed across all $d$ coordinates.
Then the social costs of the median mechanism $\SC(\locs,\median)$ and the optimum $\SC(\locs,\facility)$ are respectively
\begin{align*}
    &\sum_{\loc\in\locs} \qnorm{\m-\loc}=\frac{1}{1-c^*}\cdot\flr{a^*\cdot d}^{1/q}\cdot\frac{n}{2-2a^*}+d^{1/q}\cdot\frac{1-2a^*}{2-2a^*}\cdot n,\\
    &\sum_{\loc\in\locs} \qnorm{\facility-\loc}=\left(\left(\frac{c^*}{1-c^*}\right)^q\cdot \flr{a^*\cdot d}+(d-\flr{a^*\cdot d})\right)^\frac{1}{q}\cdot\frac{n}{1-2a^*}.
\end{align*}
We note that when $d\rightarrow\infty$, the approximation ratio for our instance is approaching $\UB(q)=1/\lambda^*$, as the difference between $\flr{a^*\cdot d}$ and $a^*\cdot d$ is negligible: 
\begin{align*}
    \frac{\sum_{\loc\in\locs} \qnorm{\m-\loc}}{\sum_{\loc\in\locs} \qnorm{\facility-\loc}}
    \underset{d\to+\infty}{\longrightarrow}
    \frac{\frac{1}{1-c^*}\cdot (a^*)^{1/q}+1-2a^*}{\left(\left(\frac{c^*}{1-c^*}\right)^q\cdot a^* +1-a^*\right)^{1/q}}=1/\lambda^*.
\end{align*}
For a given $d$ we get the following lower bound on the approximation ratio
\begin{align*}
    &\frac{\sum_{\loc\in\locs} \qnorm{\m-\loc}}{\sum_{\loc\in\locs} \qnorm{\facility-\loc}}=\frac{\frac{1}{1-c^*}\cdot\flr{a^*\cdot d}^{1/q}+d^{1/q}\cdot(1-2a^*))}{\left(\left(\frac{c^*}{1-c^*}\right)^q\cdot \flr{a^*\cdot d}+(d-\flr{a^*\cdot d})\right)^{1/q}}\\
    &\leq \frac{\frac{1}{1-c^*}\cdot(a^*\cdot d)^{1/q}+d^{1/q}\cdot(1-2a^*))}{\left(\left(\frac{c^*}{1-c^*}\right)^q\cdot (a^*\cdot d-1)+(d-a^*\cdot d)\right)^{1/q}}\\
    &\leq \frac{\frac{1}{1-c^*}\cdot(a^*\cdot d)^{1/q}+d^{1/q}\cdot(1-2a^*))}{\left(\left(\frac{c^*}{1-c^*}\right)^q\cdot (a^*\cdot d)+(d-a^*\cdot d)\right)^{1/q}}\cdot\frac{1}{(1-O(\frac{1}{d}))^{1/q}}\\
    &\leq \frac{1}{\lambda^*\cdot(1-O(\frac{1}{d}))}.
\end{align*}
For the important special case of $\lqnorm[\infty](\reals^d)$, our instance consists of the following two types of points.
\begin{description}
        \item[Type I.] One coordinate is $2$ and all remaining coordinates are $0$.        
        \item[Type II.] $\loc=(1,1,\cdots,1)$ the same as $\facility$.
\end{description}
There are $\frac{d}{2d-1}\cdot n$ points of Type I and $\frac{d-1}{2d-1}\cdot n$ points of Type II. The non zero coordinate of Type I points are distributed uniformly over $[d]$.  The approximation ratio of the median mechanism is
\begin{align*}
    \frac{\sum_{\loc\in\locs} \qnorm[\infty]{\m-\loc}}{\sum_{\loc\in\locs} \qnorm[\infty]{\facility-\loc}}
    = \frac{2\cdot\frac{d}{2d-1}\cdot n+1\cdot\frac{d-1}{2d-1}\cdot n}{1\cdot\frac{d}{2d-1}\cdot n}=3-\frac{1}{d}\underset{d\to+\infty}{\longrightarrow}
    3.
\end{align*}

In summary, we obtain the following result
\begin{theorem}
    \label{thm:LB}
   For $\lqnorm(\reals^d)$ with any $q\ge 1$ and $d\in\nats_{\ge 1}$, the 
   approximation ratio of the median mechanism is at most $\LB(q,d)=(1-O(\frac{1}{d}))\cdot\UB(q)$.
   I.e., the upper bound $\UB(q)$ is tight for $d\rightarrow\infty$.
\end{theorem}

\section{Generalized Median of~\cite{AgrawalBGTX22} in $\mathbb{L}_2(\mathbb{R}^d)$}
\label{sec:prediction}

We discuss in this section how to extend our analysis of the median mechanism to the generalized median mechanism of~\cite{AgrawalBGTX22} in strategic facility location problem with prediction. In the mechanism design setting with prediction, the mechanism $\mech$ receives  reports $\locs=(\loc_1,\loc_2,\cdots,\loc_n)$ of $n$ strategic agents along with the learning-augmented advice for the placement of the optimal facility $\pred\in\reals^d$. When prediction $\pred$ is accurate, the mechanism $\mech(\locs,\pred)$ is expected to perform better than the worst-case approximation guarantee. The approximation ratio of $\mech$ in the case $\pred=\facility(\locs)$ is called \emph{consistency guarantee}. On the other hand, the prediction $\pred$ might also be highly inaccurate. In this case, the mechanism is expected to retain some of the worst-case performance even when $\pred$ has arbitrary bad social cost compared to the optimum $\facility(\locs)$ for the actual locations of the agents. The worst-case approximation ratio of $\mech(\locs,\pred)$ over all possible predictions $\pred\in\reals^d$
is called \emph{robustness guarantee} of $\mech$.  
Agrawal, Balkanski, Gkatzelis, Ou, and Tan~\cite{AgrawalBGTX22} have shown that the following generalized median mechanism called $\CMP(c)$ parametrized with $c\in[0,1]$ achieves best possible trade-offs for the consistency and robustness in $\lqnorm[2](\reals^2)$: add $c\cdot n$ copies of the predicted point $\pred$ to the input $\locs$ and calculate coordinate-wise median 
\[
\cmp=\median(\loc_1,\loc_2,\cdots,\loc_n, c\cdot n\text{ copies of}\ \pred).
\]
We would like to understand how the consistency and robustness approximation guarantees of $\CMP(c)$ degrade in 
arbitrary Euclidean spaces $\lqnorm[2](\reals^d)$. 

\begin{theorem}
    \label{thm:consistency_guarantee}
    The consistency guarantee of $\CMP(c)$ mechanism in $\lqnorm[2](\R^d)$ is at most
    \[
    \begin{cases}
    \left( 4\sqrt{2c+3}\cdot c + 6\sqrt{2c+3} - 10c - 8 \right)^{1/2} / (c+1), & \quad\text{when }c \in [0, 1/2)\\
    \sqrt{\frac{2}{c+1}} & \quad\text{when }c \in [1/2, 1).
    \end{cases}
    \]
\end{theorem}
\begin{proof}
We make the same assumptions as in \eqref{eq:main} about optimal facility $\facility$ (that $\facility_i>0$ for each $i\in[d]$), which is the same as prediction $\facility(\locs)=\pred$, and that the generalized median $\cmp=\vect{0}$. As the constraint of $\cmp$ being located at $\mathbf{0}$ can be written as $\sum_{i\in[n]}\sigma(\loc_i)=-c\cdot n$, we get the following optimization problem
\begin{align}
 \label{eq:consistency main}
 \min\limits_{\locs}& \sum_{i\in[n]} g(\loci), &\ 
 where~ g(\loci)\eqdef\qnorm[2]{\loci-\facility}-\lambda\qnorm[2]{\loci} \nonumber\\
 \text{s.t.}&  \sum_{i\in [n]}\sigma(\loci) = -cn\cdot\vect{1}, & 
\end{align}
We follow the same plan as in our analysis of the median mechanism with minor modifications. Namely, we have the same function $g(\loc)$ as in \eqref{eq:main} defined for $q=2$; apply Lemma~\ref{lm:expression for local optimum} for each $\loci\in\locs$, to express $g(\loci)$ as a function of $x_i$ as defined in \eqref{eq:relaxed main}. 
The only modification we need to make comes from a different constraint of \eqref{eq:main}, which becomes
\begin{equation}
    \label{eq: consistency relaxed_median_constraint}
    \sum\limits_{i\in[n]}x_i=\sum\limits_{i\in[n]}\energy[S(\loci)]=
    \sum_{i\in[n]}\sum_{j:p_{i,j}\ge 0^+}f_j^2
    =\sum_{j\in[d]}f_j^2\cdot\sum_{i:p_{i,j}\ge 0^+}1
    =\frac{n}{2}\cdot\sum_{j\in[d]}f_j^2=\frac{1-c}{2}\cdot n.
\end{equation}
Thus, we get the following analog of~\eqref{eq:relaxed main} for $q=2$.
\begin{align}
 \label{eq:consistency relaxed main}
 \min\limits_{\vx}& \sum_{i\in[n]} h(x_i)
 &\text{where }h(x_i)
 \eqdef \lambda\cdot\left(\delta\cdot(1-x_i)^{1/2}- x_i^{1/2}\right)\nonumber\\
 \text{s.t.}&  \sum_{i\in [n]}x_i = \frac{1-c}{2}\cdot n,&
\forall i\in[n]~x_i\in[0,1]
\end{align}
As before, $z\eqdef\delta^{-2/3}/(\delta^{-2/3}+1)<1/2$ for $q=2$; and by Lemma~\ref{optimal value of x_i}, the optimal value of $x_i$ lies in $\{ a,1\}$, where $a\in[0,z]$ and the space of optimal solution can be narrowed down to 2 parameters $a$ and $|\{i: x_i=a\}|$\footnote{Actually, Lemma~\ref{optimal value of x_i} helps to narrow down the optimal value of $x_i$ to 3 possible values $\{a,b,1\}$, but with the same argument as for the median mechanism, we can ignore the effect of $x_i=b$.}. 

As our new constraint is $\sum_{i=1}^n x_i = (1-c)\cdot n/2$, we get that the numbers $|\{i: x_i=a\}|$ and $|\{i: x_i=1\}|$ must be, respectively, $\frac{1+c}{2(1-a)}\cdot n$ and $\frac{1-2a-c}{2(1-a)}\cdot n$. Moreover, the number $|\{i: x_i=1\}|\geq 0$ implies that $1-2a-c\geq 0$. Therefore, we need to find $\lambda$ such that
\begin{multline}
    \label{eq: consistency optimization over a}   \min_{a\in\left[0,\min\{z,\frac{1-c}{2}\}\right]}\left[\frac{1+c}{2(1-a)}\cdot n\cdot h(a) + \frac{1-2a-c}{2(1-a)}\cdot n\cdot h(1)\right]\\
    =\min_{a\in\left[0,\min\{z,\frac{1-c}{2}\}\right]}\bigg[\frac{n\cdot\lambda}{2(1-a)}\cdot\Big((1+c)\cdot\left(\delta\cdot(1-a)^{1/q}-a^{1/q}\right)-1+2\cdot a+c\Big)\bigg]\ge 0.
\end{multline}
Thus, we need to minimize $$u_1(a)\eqdef(1+c)\cdot\left(\delta\cdot(1-a)^{1/q}-a^{1/q}\right)-1+2\cdot a+c.$$
Let $a_1\eqdef \argmin\limits_{a\in\left[0,\min\{z,\frac{1-c}{2}\}\right]} u_1(a)$. We observe that
\begin{lemma}
    \label{lm:consistency optimal solution for a_1}
    If $u_1(a_1)=0$, then $a_1$ is optimal $\Rightarrow u_1'(a_1)=0\ or\ a_1=\frac{1-c}{2}$.
\end{lemma}
The proof of Lemma~\ref{lm:consistency optimal solution for a_1} is 
essentially the same as the proof of Lemma~\ref{lm:optimal solution for a}:
we simply verify that if $u_1(a_1)=0$, $a_1\in\{0,z\}$ are not optimal, and that $u_1'(a)=0$ has unique 
solution on $(0,z)$ (the only difference is that we have an extra constraint $a_1\le \frac{1-c}{2}$, which is appropriately reflected in the statement of Lemma~\ref{lm:consistency optimal solution for a_1}).
Proof with full details is deferred to the Appendix.

To conclude the proof of Theorem~\ref{thm:consistency_guarantee}, we need to find $\lambda_1$ (analog of $\lambda^*$ from the proof of Theorem~\ref{thm:UB}) such that the minimum of \eqref{eq:consistency relaxed main} and respectively \eqref{eq: consistency optimization over a} is equal to $0$. This will ensure an $1/\lambda_1$-consistency guarantee. By Lemma~\ref{lm:consistency optimal solution for a_1}, the optimal $\lambda_1$ and $a_1$ must satisfy: (i) $u_1(a_1)=0$ and (ii) $u_1'(a_1)=0$ or $a_1=(1-c)/2$. Then by denoting $\delta_1=\delta(\lambda_1)$, we find  $a_1$, $\delta_1$, $\lambda_1$ in the same way as before. We move detailed calculations to Appendix and conclude that
\begin{equation*}
    \lambda_1=\left\{
    \begin{aligned}
    &(c+1)/\left(4\sqrt{2c+3}\cdot c+6\sqrt{2c+3}-10c-8\right)^{\frac{1}{2}},&c\in[0,1/2)\\
    &\sqrt{\frac{c+1}{2}},&c\in[1/2,1).
    \end{aligned}
    \right.
\end{equation*}
\end{proof}

In the robustness scenario, instead of having $\pred=\facility$, the prediction $\pred$ can be anywhere in $\reals^d$. Thus, we first need to understand what is the worst-case placement of $\pred$ for $\cmp(\locs,\pred)$. Now, if we normalize the problem so that $\cmp=\vect{0}$, our constraint in \eqref{eq:main} becomes $\sum_{i\in[n]}\sigma(\loc_i)=-cn\cdot\sigma(\pred)$. Since this constraint only depends on the signature of $\pred$, the respective mathematical program for the robustness is as follows.
\begin{align}
 \label{eq:optimization for sigma}
 \min\limits_{\locs,\sigma(\pred)}&\sum_{i\in[n]} g(\loci) &\text{where }g(\loci)=\qnorm[2]{\loci-\facility}-\lambda\cdot\qnorm[2]{\loci} \nonumber\\
 \text{s.t.}&  \sum_{i\in [n]}\sigma(\loci) = -cn\cdot\sigma(\pred),& 
\end{align}
Let $\sigma^*$ and $\loc^*$ be the optimal solution to \eqref{eq:optimization for sigma}. It is rather straightforward to see that $\sigma^*=\sigma(\pred)=(-1,\ldots,-1)$.
\begin{lemma} 
    \label{lm:optimal solution for sigma(pred)}
    Without loss of generality we can assume that 
    $\sigma^*=(-1,-1,\cdots,-1)$.
\end{lemma}
\begin{proof} Assume towards the contradiction that there is a coordinate $k$ with $\sigma^*_k=1$. Then we shall change $\sigma^*_k$ from $1$ to $-1$ ($\sigma'=(\sigma^*_{\text{-}k},-1)$) and also modify the set of locations to $\locs'$, such that $(\locs',\sigma')$ satisfy the constraint and achieve even smaller objective value $\sum_{i\in[n]}g(\loci')\le\sum_{i\in[n]}g(\loci)$. 

We construct $\locs'$ as follows. For any $j\neq k$ and $i\in[n]$, let $p'_{i,j}=p^*_{i,j}$. For coordinate $k$, we select any $cn$ locations $\loc^*_{(1)},\loc^*_{(2)},\cdots,\loc^*_{(cn)} \in \locs^*$ with the signature of the $k$-th coordinate being -1 and let $p'_{(i),k}=-p^*_{(i),k}$. The $k$-th coordinate remains the same between $\locs'$ and $\locs^*$ for the rest $(1-c)\cdot n$ locations. Now, since $f_j>0$ for all $j\in d$, $|p'_{(i),k}-f_k|\leq|p^*_{(i),k}-f_k|$. Therefore, 
    \begin{align*}
        \sum_{i\in[n]} g(\loci')&=\sum_{i\in[n]}\qnorm[2]{\loci'-\facility}-\lambda\cdot\qnorm[2]{\loci'}\\
        &\leq \sum_{i\in[n]}\qnorm[2]{\loci^*-\facility}-\lambda\cdot\qnorm[2]{\loci^*}\\
        &=\sum_{i\in[n]} g(\loci^*).
    \end{align*}
     I.e., $\locs'$ produce a smaller or equal objective value.
\end{proof}
Hence, \eqref{eq:optimization for sigma} can be written as
\begin{align}
 \label{eq:robustness main}
\min\limits_{\locs}& \sum_{i\in[n]} g(\loci), &\ 
 where~ g(\loci)\eqdef\qnorm[2]{\loci-\facility}-\lambda\qnorm[2]{\loci} \nonumber\\
 \text{s.t.}&  \sum_{i\in [n]}\sigma(\loci) = cn\cdot\vect{1}, &
\end{align}
Notice that the optimization problem~\eqref{eq:robustness main} is very similar to~\eqref{eq:consistency main}. We implement the same reasoning as in the proof of Theorem~\eqref{thm:consistency_guarantee} and obtain the following robustness guarantee 
\begin{theorem}
    \label{thm:robustness_guarantee}
    The robustness guarantee of $\CMP(c)$ mechanism in $\lqnorm[2](\R^d)$ is at most
    \begin{align*}
    \left( -4\sqrt{3-2c}\cdot c + 6\sqrt{3-2c} + 10c - 8 \right)^{1/2} / (1-c).
    \end{align*}
\end{theorem}

 We defer the detailed proof to Appendix.

\subsection*{Comparison between $\CMP(c)$ in $\reals^d$ and $\reals^2$}
The consistency and robustness approximation guarantees of the coordinate-wise median mechanism from~\cite{AgrawalBGTX22} are respectively $\sqrt{2c^2+2}/(c+1)$ and $\sqrt{2c^2+2}/(1-c)$ in $\reals^2$. We compare these results to the performance of $\CMP(c)$ in $\lqnorm[2](\reals^d)$ for arbitrary $d$, given by our Theorems~\ref{thm:consistency_guarantee} and~\ref{thm:robustness_guarantee}.

In figure~\ref{fig:plot of $r_a$ and $r_b$}, we plot two curves corresponding to the ratio $r_a\eqdef$ consistency guarantee in $\reals^d$/consistency guarantee in $\reals^2$ and to the ratio $r_b\eqdef$ robustness guarantee in $\reals^d$/robustness guarantee in $\reals^2$ depending on the parameter $c\in[0,1)$. Specifically,

\begin{equation*}
r_a=
\left\{
\begin{aligned}
&(2\sqrt{2c+3}\cdot c+3\sqrt{2c+3}-5c-4)^{\frac{1}{2}}/(c^2+1)^{\frac{1}{2}}, &c\in[0,1/2)\\
&(c+1)^{\frac{1}{2}}/(c^2+1)^{\frac{1}{2}},  &c\in[1/2,1)
\end{aligned}
\right.
\end{equation*}
and,
\begin{equation*}
r_b=(-2\sqrt{3-2c}\cdot c+3\sqrt{3-2c}+5c-4)^{\frac{1}{2}}/(c^2+1)^{\frac{1}{2}}.
\end{equation*}
When $c=0$, $r_a=r_b=\sqrt{6\sqrt{3}-8}/\sqrt{2}\approx1.09$. From figure~\ref{fig:plot of $r_a$ and $r_b$}, $r_a$ increases at first and then decreases to $1$, but it never exceeds $1.11$, while $r_b$ decreases monotonously from roughly $1.094$ to $1$. This indicates that the generalization of the CMP mechanism from $\R^2$ to $\R^d$ won't suffer a performance loss (both for consistency and for robustness) more than $11\%$.
\begin{figure}[htb] 
		\centering
		\includegraphics[width=5.0in]{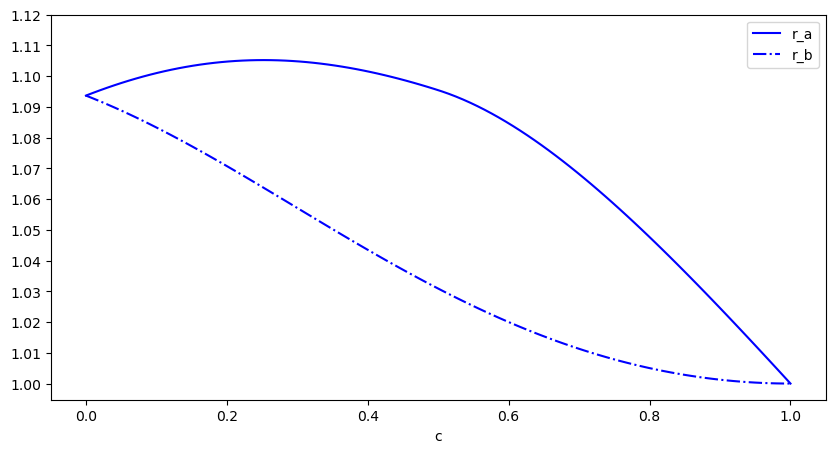}
		\caption{plot of $r_a$ and $r_b$}
\label{fig:plot of $r_a$ and $r_b$}		
\end{figure}

On the other hand, when comparing asymptotic convergence rate of both consistency and robustness for confidence parameter $c$ approaching $c\rightarrow 1$ of $\CMP(c)$ in $\reals^d$ against $\CMP(c)$ in $\reals^2$, we observe that robustness guarantees are essentially the same, while consistency guarantees converge to $1$ at different rates. Specifically, both robustness guarantees are of order $\Theta(1/(1-c))$, and the Taylor expansions of the consistency guarantee for $\eps=1-c$ is  $\sqrt{2/(c+1)}=1+\frac{\eps}{4}+O(\eps^2)$ in $\reals^d$, and 
$\sqrt{2c^2+2}/(c+1)=1+\frac{\eps^2}{8}+O(\eps^3)$ in $\reals^2$.


\section{Conclusions}
\label{sec:conclusions}

In this work, we derived tight approximation ratios for the coordinate-wise median mechanism in $\lqnorm(\reals^d)$ for all $q\ge 1$ and unbounded dimension $d\to+\infty$.
We found that these approximation ratios are all constants never exceeding $3$, which in Euclidean space $\lqnorm[2]$ are only slightly larger than the respective approximation ratio in low-dimensional space ($1.55$ for $\lqnorm[2](\reals^d)$ versus $1.41$ for $\lqnorm[2](\reals^2)$). Hence, our work implies that there is no reason from the perspective of approximation guarantees to restrict attention to $\lqnorm[2](\reals^2)$ case in strategic facility location setting. On the other hand, there is still a lack of satisfying characterization of strategy-proof mechanism in higher-dimensions $d\ge 3$.

Namely, the existing characterizations that work for general $d$-dimensional spaces ~\cite{border1983straightforward,peremans1997strategy} allow agents to have different  quadratic separable preferences, which places more restriction on the space of possible strategy-proof mechanisms than in the strategic facility location setting, where agents' costs are defined by the same distance function. I.e., there could be other strategy-proof anonymous mechanisms than generalized median. Obtaining such  characterization for any space $\lqnorm(\reals^d)$ seems an interesting question.





\bibliographystyle{abbrv}
\bibliography{facility}

\appendix

\section{Missing proofs}
\label{sec:appendix}
\subsection{Proof of Lemma~\ref{lm:expression for local optimum}}
\begin{proof}
For $S(\loc)\neq\emptyset$, according to Lemma~\ref{lm:local optimal condition}, we get $p_j=f_j/(1-c)$ for $j\in S$. Hence, $\qnorm{\loc-\facility}$ and $\qnorm{\loc}$ can be expressed as functions of $c$:
\begin{equation*}
\begin{split}
&\qnorm{\loc-\facility}=\left(\sum_{j\in S}(p_j-f_j)^q+\sum_{j\in \overline{S}}f_j^q\right)^{1/q}=\left(\left(\frac{c}{1-c}\right)^q\energy+\energyc\right)^{1/q}\\
&\qnorm{\loc}=\left(\sum_{j\in S}p_j^q\right)^{1/q}=\frac{1}{1-c}\cdot\energy^{1/q}.
\end{split}
\end{equation*}
Recall that $c=\lambda^{1/(q-1)}\cdot\qnorm{\loc-\facility}/\qnorm{\loc}$, we now derive the following equations that allow us to eliminate dependency on $c$ of $g(\loc)$:
\begin{align*}
&c=\lambda^{1/(q-1)}\cdot\frac{\left((\frac{c}{1-c})^q\cdot\energy+\energyc\right)^{1/q}}{\frac{1}{1-c}\cdot\esq}=\lambda^{1/(q-1)}\cdot\left(c^q+(1-c)^q\cdot\frac{\energyc}{\energy}\right)^{1/q}\\
&\Leftrightarrow c^q=\lambda^{q/(q-1)}\cdot\left(c^q+(1-c)^q\cdot\frac{\energyc}{\energy}\right)\\
&\Leftrightarrow c^q\cdot\left(1-\lambda^{q/(q-1)}\right)=\lambda^{q/(q-1)}\cdot(1-c)^q\cdot\frac{\energyc}{\energy}\\
&\Leftrightarrow\frac{c}{1-c}=\left(\frac{\energyc}{\energy}\right)^{1/q}\cdot\left(\frac{\lambda^{q/(q-1)}}{1-\lambda^{q/(q-1)}}\right)^{1/q}\\
&\Leftrightarrow\frac{1}{1-c}=\left(\frac{\energyc}{\energy}\right)^{1/q}\cdot\left(\frac{\lambda^{q/(q-1)}}{1-\lambda^{q/(q-1)}}\right)^{1/q}+1.
\end{align*}
Therefore,
\begin{align*}
g(\loc)&=\qnorm{\facility-\loc}-\lambda\cdot\qnorm{\loc}\\
&=\left(\left(\frac{c}{1-c}\right)^q\cdot\energy+\energyc\right)^{1/q}-\frac{1}{1-c}\cdot\energy^{1/q}\\
&=\left(\frac{\lambda^{q/(q-1)}}{1-\lambda^{q/(q-1)}}\cdot\energyc+\energyc\right)^{1/q}-\lambda\cdot\left(\left(\frac{\energyc}{\energy}\right)^{1/q}\cdot\left(\frac{\lambda^{q/(q-1)}}{1-\lambda^{q/(q-1)}}\right)^{1/q}+1\right)\cdot \energy^{1/q}\\
&=\left(\frac{1}{1-\lambda^{q/(q-1)}}\cdot\energyc\right)^{1/q}-\lambda^{q/(q-1)}\cdot \energyc^{1/q}\cdot\left(\frac{1}{1-\lambda^{q/(q-1)}}\right)^{1/q}-\lambda\cdot\energy^{1/q}\\
&=\left(1-\lambda^{q/(q-1)}\right)^{(q-1)/q}\cdot\ecsq-\lambda\cdot\esq.
\end{align*}
For $S(\loc)=\emptyset$, we first consider the situation when $T=\emptyset$. In this case, $\loc=\zerovec$, which is not local optimum. When $T\neq\emptyset$, for a fixed $T$, according to Lemma~\ref{lm:local optimal condition}, $(f_j-p_j)/(-p_j)=c$ for $j\in T$, which is equivalent to $-p_j=f_j/(c-1)$. Hence, $\qnorm{\facility-\loc}$ and $\qnorm{\loc}$ can be expressed as functions of c:
\begin{align*}
&\qnorm{\facility-\loc}=\left(\sum_{j\in T}(f_j-p_j)^q+\sum_{j\in \overline{T}}f_j^q\right)^{1/q}=\left(\left(\frac{c}{c-1}\right)^q\energyT+\energyTc\right)^{1/q}\\
&\qnorm{\loc}=\left(\sum_{j\in T}(-p_j)^q\right)^{1/q}=\frac{1}{c-1}\cdot\energyT^{1/q}.
\end{align*}
Recall that $c=\lambda^{1/(q-1)}\cdot\qnorm{\loc-\facility}/\qnorm{\loc}$, we now derive the following equations that allow us to eliminate dependency on $c$ of $g(\loc)$:
\begin{align*}
&c=\lambda^{1/(q-1)}\cdot\frac{\left((\frac{c}{c-1})^q\cdot\energyT+\energyTc\right)^{1/q}}{\frac{1}{c-1}\cdot\energyT^{1/q}}=\lambda^{1/(q-1)}\cdot\left(c^q+(c-1)^q\cdot\frac{\energyTc}{\energyT}\right)^{1/q}\\
&\Leftrightarrow c^q=\lambda^{q/(q-1)}\cdot\left(c^q+(c-1)^q\cdot\frac{\energyTc}{\energyT}\right)\\
&\Leftrightarrow c^q\cdot\left(1-\lambda^{q/(q-1)}\right)=\lambda^{q/(q-1)}\cdot(c-1)^q\cdot\frac{\energyTc}{\energyT}\\
&\Leftrightarrow\frac{c}{c-1}=\left(\frac{\energyTc}{\energyT}\right)^{1/q}\cdot\left(\frac{\lambda^{q/(q-1)}}{1-\lambda^{q/(q-1)}}\right)^{1/q}\\
&\Leftrightarrow\frac{1}{c-1}=\left(\frac{\energyTc}{\energyT}\right)^{1/q}\cdot\left(\frac{\lambda^{q/{q-1}}}{1-\lambda^{q/(q-1)}}\right)^{1/q}-1.
\end{align*}
Therefore,
\begin{align*}
g(\loc)&=\qnorm{\facility-\loc}-\lambda\cdot\qnorm{\loc}\\
&=\left(\left(\frac{c}{1-c}\right)^q\cdot\energyT+\energyTc\right)^{1/q}-\frac{1}{1-c}\cdot\energyT^{1/q}\\
&=\left(\frac{\lambda^{q/(q-1)}}{1-\lambda^{q/(q-1)}}\cdot\energyTc+\energyTc\right)^{1/q}-\lambda\cdot\left(\left(\frac{\energyTc}{\energyT}\right)^{1/q}\cdot\left(\frac{\lambda^{q/(q-1)}}{1-\lambda^{q/(q-1)}}\right)^{1/q}-1\right)\cdot \energyT^{1/q}\\
&=\left(\frac{1}{1-\lambda^{q/(q-1)}}\cdot\energyTc\right)^{1/q}-\lambda^{q/(q-1)}\cdot \energyTc^{1/q}\cdot\left(\frac{1}{1-\lambda^{q/(q-1)}}\right)^{1/q}+\lambda\cdot\energyT^{1/q}\\
&=\left(1-\lambda^{q/(q-1)}\right)^{(q-1)/q}\cdot\ectq+\lambda\cdot\etq.
\end{align*}
Then, we minimize over all possible $T$ to solve for $g(\loc)$ under $T\neq\emptyset$:
\begin{align*}
g(\loc)=\min\limits_{\substack{T\subseteq[d]\\T\neq\emptyset}}\left[\left(1-\lambda^{q/(q-1)}\right)^{(q-1)/q}\cdot\ectq+\lambda\cdot\etq\right].
\end{align*}
We view $\left(1-\lambda^{q/(q-1)}\right)^{(q-1)/q}\cdot\ectq+\lambda\cdot\etq$ as a function of $\energyT$ and, to simplify notation denote $x\eqdef\energyT$. First, we figure out the domain of $x$. Since $c>1$ from Lemma~\ref{lm:local optimal condition}, we have:
\begin{align*}
    \left(\frac{\energyTc}{\energyT}\right)^{1/q}\cdot\left(\frac{\lambda^{q/(q-1)}}{1-\lambda^{q/(q-1)}}\right)^{1/q}=\frac{c}{c-1}>1
\Leftrightarrow\quad \frac{\energyTc}{\energyT}\cdot\frac{\lambda^{q/(q-1)}}{1-\lambda^{q/(q-1)}}>1
\quad\Leftrightarrow\quad \energyT<\lambda^{q/q-1}.
\end{align*}
Together with $\energyT>0$, we have $x=\energyT\in(0,\lambda^{q/q-1})$. Then let's consider function
\begin{align*}
    t(x)=\left(1-\lambda^{q/(q-1)}\right)^{(q-1)/q}\cdot(1-x)^{1/q}+\lambda\cdot x^{1/q}.
\end{align*}
Since both $(1-x)^{1/q}$ and $x^{1/q}$ are concave in $x$, $t(x)$ (as a non-negative weighted sum of $(1-x)^{1/q}$ and $x^{1/q}$) is also a concave function of $x$. Hence, its minimum value is achieved either when $x$ goes to $0$, or when $x$ goes to $\lambda^{q/q-1}$. As
\begin{align*}
    t(x)\left\vert_{x\rightarrow 0}=1-\lambda^{q/(q-1)}<t(x)\right\vert_{x\rightarrow \lambda^{q/q-1}}=1, 
\end{align*}
we conclude that $g(\loc)\geq1-\lambda^{q/(q-1)}$ when $T\neq\emptyset$.
\end{proof}

\subsection{Proof of Lemma~\ref{lm:optimal solution for a}}
\begin{proof}
The optimal condition for $u(a)$ is that either $a^*\in\{0,z\}$ or $u'(a)=0$. Specifically,
\begin{align*}
   u'(a)=\frac{1}{q}\cdot\left(-\delta\cdot(1-a)^{(1-q)/q}-a^{(1-q)/q}\right)+2.
\end{align*}
First, notice that $u'(a)\rightarrow -\infty$ when $a\rightarrow 0$, saying $a=0$ is not optimal.
Then, we observe that $u'(a)=0\Leftrightarrow w(a)\eqdef\frac{1}{q}\cdot\left(-\delta\cdot(1-a)^{(1-q)/q}-a^{(1-q)/q}\right)+2=0$. We prove that there exists at most one solution to $u'(a)=0$ by showing that $w(a)$ is strictly increasing in $a$ on $a\in(0,z)$:
\begin{align*}
  w'(a)&=\frac{1-q}{q^2}\cdot(1-a)\cdot\left(\delta\cdot(1-a)^{(1-2q)/q}-a^{(1-2q)/q}\right) \\
  &= \frac{q-1}{q^2}\cdot(1-a)\cdot\left(-\delta\cdot(1-a)^{(1-2q)/q}+a^{(1-2q)/q}\right)> 0.
\end{align*}
The last inequality follows from the fact that $-\delta\cdot(1-x)^{(1-2q)/q}+x^{(1-2q)/q}$ is convex in $x$ when $x\in(0,z)$.
On the other hand, we prove that there exists one unique solution to $u'(a)=0$ on $a\in(0,z)$ by contradiction. We assume that there is no solution for $u'(a)=0$ on $a\in(0,z)$. I.e., $u'(z)<0$ and therefore $a^*=z$. Now, since we have $u(a^*)=0$, we get the following system of inequalities on $\delta$ and $z$.
\begin{equation}
\label{eq:z}
\left\{
    \begin{aligned}
    &\delta\cdot(1-z)^{1/q}-{z}^{1/q}-1+2\cdot z=0\\
    &\frac{1}{q}\cdot\left(-\delta\cdot(1-z)^{(1-q)/q}-z^{(1-q)/q}\right)+2<0.
    \end{aligned}
    \right.
\end{equation}
We multiply $(1-z)$ to both side of the second inequality of \eqref{eq:z}, which gives us
\begin{equation}
\label{eq:inequality of z}
    \frac{1}{q}\cdot\left(-\delta\cdot(1-z)^{1/q}+z^{1/q}-z^{(1-q)/q}\right)+2-2z<0.
\end{equation}
By adding the first equation in \eqref{eq:z} to both side of \eqref{eq:inequality of z}, we get
\begin{equation}
\label{eq:inequality of z(2)}
    \left(1-\frac{1}{q}\right)\cdot\left(\delta\cdot(1-z)^{1/q}-z^{1/q}\right)-\frac{1}{q}\cdot z^{(1-q)/q}+1<0.
\end{equation}
However,
\begin{align*}
   &\left(1-\frac{1}{q}\right)\cdot\left(\delta\cdot(1-z)^{1/q}-z^{1/q}\right)-\frac{1}{q}\cdot z^{(1-q)/q}+1\\
   &=\left(1-\frac{1}{q}\right)\cdot\left(z^{(1-2q)/q}\cdot(1-z)^2-z^{1/q}\right)-\frac{1}{q}\cdot z^{(1-q)/q}+1\\
   &=z^{(1-2q)/q}\cdot\left(\left(1-\frac{1}{q}\right)\cdot\left((1-z)^2-z^2\right)-\frac{1}{q}\cdot z\right)+1\\
   &=z^{(1-2q)/q}\cdot\left(1-2z-(1-z)\cdot\frac{1}{q}\right)+1\\
    &\ge -z^{(1-q)/q}\cdot\frac{1}{2}+1\\
   &\ge -\frac{1}{2z}+1\\
   &\geq 0.
\end{align*}
The equation in the second line follows from the fact that $\delta\cdot(1-z)^{(1-2q)/q}=z^{(1-2q)/q}$. The inequalities in the 5-th, 6-th, and 7-th lines hold, since $q\ge 1$ and $z\leq1/2$. Hence, there exists exactly one point with $u'(a)=0$ when $a\in [0,z]$. 
\end{proof}

\subsection{Proof of Lemma~\ref{lm:consistency optimal solution for a_1}}
\begin{proof}
The optimal condition for $u_1(a)$ is that $$a_1\in\{0,z,
\frac{1-c}{2}\}\ or~ u_1'(a_1)=0.$$ Specifically,
\begin{align*}
   u_1'(a)=\frac{1}{2}\cdot(1+c)\cdot\left(-\delta\cdot(1-a)^{-1/2}-a^{-1/2}\right)+2.
\end{align*}
First, notice that $u_1'(a)\rightarrow -\infty$ when $a\rightarrow 0$, showing that $a=0$ is not optimal.

Then, we observe that $u_1'(a)=0\Leftrightarrow w_1(a)\eqdef\frac{1}{2}\cdot(1+c)\cdot\left(-\delta\cdot(1-a)^{-1/2}-a^{-1/2}\right)+2=0$. We prove that there exists at most one solution to $u_1'(a)=0$ by showing that $w_1(a)$ is strictly increasing in $a$ on $a\in(0,z)$:
\begin{align*}
  w_1'(a)&=-\frac{1}{4}\cdot(1+c)\cdot(1-a)\cdot\left(\delta\cdot(1-a)^{-3/2}-a^{-3/2}\right) \\
  &= \frac{1}{4}\cdot(1+c)\cdot(1-a)\cdot\left(-\delta\cdot(1-a)^{-3/2}+a^{-3/2}\right)> 0.
\end{align*}
The last inequality follows from the fact that $-\delta\cdot(1-x)^{-3/2}+x^{-3/2}$ is convex in $x$ when $x\in(0,z)$.
On the other hand, we prove that the optimal solution can never be $a_1=z$ by contradiction. 

Assume that $a_1=z$, which implies that $u_1'(a)<0$ on $a\in[0,z)$. Since we have $u_1(a_1)=0$, we can get the following system of inequalities on $\delta$ and $z$
\begin{equation}
\label{eq:z for consistency}
\left\{
    \begin{aligned}
    &(1+c)\cdot(\delta\cdot(1-z)^{1/2}-{z}^{1/q})-1+2\cdot z+c=0\\
    &\frac{1}{2}\cdot(1+c)\cdot\left(-\delta\cdot(1-z)^{-1/2}-z^{-1/2}\right)+2<0.
    \end{aligned}
    \right.
\end{equation}
We multiply $(1-z)$ to both side of the second inequality of \eqref{eq:z for consistency}, which gives us
\begin{equation}
\label{eq:consistency inequality of z}
    \frac{1}{2}\cdot(1+c)\cdot\left(-\delta\cdot(1-z)^{1/2}+z^{1/2}-z^{-1/2}\right)+2-2z<0.
\end{equation}
By adding the first equation in \eqref{eq:z for consistency} to both side of \eqref{eq:consistency inequality of z}, we get
\begin{equation}
\label{eq:consistency inequality of z(2)}
    (1+c)\cdot\left(\frac{1}{2}\cdot\left(\delta\cdot(1-z)^{1/2}-z^{1/2}\right)-\frac{1}{2}\cdot z^{-1/2}+1\right)<0.
\end{equation}
However,
\begin{align*}
   &\frac{1}{2}\cdot\left(\delta\cdot(1-z)^{1/2}-z^{1/2}\right)-\frac{1}{2}\cdot z^{-1/2}+1\\
   =&\frac{1}{2}\cdot\left(z^{-3/2}\cdot(1-z)^2-z^{1/2}\right)-\frac{1}{2}\cdot z^{-1/2}+1\\
   =&z^{-3/2}\cdot\left(\frac{1}{2}\cdot\left((1-z)^2-z^2\right)-\frac{1}{2}\cdot z\right)+1\\
   =&z^{-3/2}\cdot\left(1-2z-(1-z)\cdot\frac{1}{2}\right)+1\\
   \ge& -z^{-1/2}\cdot\frac{1}{2}+1\\
   \ge& -\frac{1}{2z}+1\\
   \geq& 0.
\end{align*}
The equation in the second line follows from the fact that $\delta\cdot(1-z)^{-3/2}=z^{-3/2}$. The inequalities in the 5-th, 6-th, and 7-th lines hold, since $q\ge 1$ and $z\leq1/2$. Hence, we rule out the possibility of $a_1\in\{0,z\}$. I.e., either there is a unique $a_1$ on $a\in\left[0,min\{z,\frac{1-c}{2}\}\right]$ such that $u_1'(a)=0$ when $a=a_1$, or $a_1=\frac{1-c}{2}$.
\end{proof}

\subsection{Missing Proof of Theorem~\ref{thm:consistency_guarantee}}
\begin{proof}
First, we compute for $a'_1$ such that $u_1'(a'_1)=0$ and $u_1(a'_1)=0$. The corresponding $a'_1$ is the optimal solution (i.e. $a_1=a'_1$) if $a'_1\leq \frac{1-c}{2}$, or otherwise, $a_1=\frac{1-c}{2}$. The equations for $u_1'(a'_1)=0$ and $u_1(a'_1)=0$ are
\begin{equation}
\label{eq:equations for a'_1}
\left\{
    \begin{aligned}
    &\left(1+c\right)\cdot\left(\delta_1\cdot(1-a'_1)^{1/2}-{(a'_1)}^{1/2}\right)-1+2\cdot a'_1+c=0\\
    &\frac{1}{2}\cdot\left(1+c\right)\cdot\left(\delta_1\cdot(1-a'_1)^{-1/2}-(a'_1)^{-1/2}\right)+2=0.
    \end{aligned}
    \right.
\end{equation}
We multiply $2(a'_1-1)$ on both side of the second equation in~\eqref{eq:equations for a'_1} and get
\begin{equation}
\label{eq:second equation for a'_1}
    \left(1+c\right)\cdot\left(\delta_1\cdot(1-a'_1)^{1/2}-(a'_1)^{1/2}+(a'_1)^{-1/2}\right)+4(a'_1-1)=0.
\end{equation}
The first equation gives us $(1+c)\cdot(\delta_1\cdot(1-a'_1)^{1/2}-(a'_1)^{1/2})=1-2\cdot a'_1+c$. By plugging it in equation~\eqref{eq:second equation for a'_1}, we get
\begin{equation}
\label{eq:third equation for a'_1}
    2\cdot a'_1+(1+c)\cdot (a'_1)^{-1/2}-3-c=0.
\end{equation}
Let $t_1=(a'_1)^{1/2}$, we can derive
\begin{align*}
&2t_1^2+(1+c)/ t_1-3-c=0\\
\Leftrightarrow& 2t_1^3-(3+c)\cdot t_1 +(1+c)=0\\
\Leftrightarrow& (t_1-1)\cdot\left(2t_1^2+2t_1-(1+c)\right)=0\\
\Leftrightarrow& t_1 = \frac{-1+\sqrt{3+2c}}{2}.
\end{align*}
Therefore, $a'_1=t_1^2=(2+c-\sqrt{3+2c})/2$. We compare $a'_1$ with $\frac{1-c}{2}$, which gives us $a_1=a'_1=t_1^2=(2+c-\sqrt{3+2c})/2$ when $c\in[0,1/2)$ and $a_1=(1-c)/2$ when $c\in[1/2,1)$. Next, by plugging the value of $a_1$ into $u_1(a_1)=0$, we can obtain the value of $\delta_1$ and $\lambda_1$. 

When $c\in[0,1/2)$, we have
\begin{align*}
  \delta_1^2&=\left(\frac{1-2a_1-c}{1+c}+a_1^{1/2}\right)^2/(1-a_1)\\
  &=\frac{1}{(1+c)^2}\cdot\left(1-2a_1-c+a_1^{1/2}\cdot(1+c)\right)^2/(1-a_1) \\
  &=\frac{1}{(1+c)^2}\cdot\left(\frac{-3-5c+(3+c)\cdot\sqrt{3+2c}}{2}\right)^2/\left(\frac{-c+\sqrt{3+2c}}{2}\right) \\
  &=\frac{2c^3+40c^2+66c+36-(18+36c+10c^2)\cdot\sqrt{3+2c}}{-2c\cdot(1+c)^2+2(1+c)^2\cdot\sqrt{3+2c}} \\
  &=\frac{c^3+20c^2+33c+18-(9+18c+5c^2)\cdot\sqrt{3+2c}}{-c\cdot(1+c)^2+(1+c)^2\cdot\sqrt{3+2c}}.
\end{align*}
Therefore,
\begin{align*}
  \lambda_1&=(1+\delta_1^2)^{-1/2}\\
  &=\left(1+\frac{c^3+20c^2+33c+18-(9+18c+5c^2)\cdot\sqrt{3+2c}}{-c\cdot(1+c)^2+(1+c)^2\cdot\sqrt{3+2c}}\right)^{-1/2}\\
  &=\left(\frac{18c^2+32c+18-(4c^2+16c+8)\cdot\sqrt{3+2c}}{-c\cdot(c+1)^2+(1+c)^2\cdot\sqrt{3+2c}}\right)^{-1/2}\\
  &=(c+1)\cdot\left(4\sqrt{2c+3}\cdot c+6\sqrt{2c+3}-10c-8\right)^{-1/2}.
\end{align*}

When $c\in[1/2,1)$, we have
\begin{align*}
  \delta_1^2&=\left(\frac{1-2a_1-c}{1+c}+a_1^{1/2}\right)^2/(1-a_1)
  =\frac{a_1}{1-a_1}
  =\frac{1-c}{1+c},
\end{align*}
and
\begin{align*}
  \lambda_1=\left(1+\delta_1^2\right)^{-1/2}=\sqrt{\frac{c+1}{2}}.
\end{align*}
\end{proof}

\subsection{Proof of Theorem~\ref{thm:robustness_guarantee}}
\begin{proof}

We follow the same reasoning as in the proof of Theorem~\ref{thm:consistency_guarantee}. First, we apply Lemma~\ref{lm:expression for local optimum} to express $g(\loci)$ as a function of $x_i$ (defined in~\eqref{eq:relaxed main}) for each $\loci\in\locs$. 
Then, we modify the constraint of~\eqref{eq:main} into
\begin{align}
    \label{eq: robustness relaxed_median_constraint}
    \sum\limits_{i\in[n]}x_i=\sum\limits_{i\in[n]}\energy[S(\loci)]=
    \sum_{i\in[n]}\sum_{j:p_{i,j}\ge 0^+}f_j^2
    =\sum_{j\in[d]}f_j^2\cdot\sum_{i:p_{i,j}\ge 0^+}1
    =\frac{n}{2}\cdot\sum_{j\in[d]}f_j^2=\frac{1+c}{2}\cdot n.
\end{align}
Thus, we get the following analog of~\eqref{eq:relaxed main} for q=2. 
\begin{align}
 \label{eq:robustness relaxed main}
 \min\limits_{\vx}& \sum_{i\in[n]} h(x_i)
 &\text{where }h(x_i)
 \eqdef \lambda\cdot\left(\delta\cdot(1-x_i)^{1/2}- x_i^{1/2}\right)\nonumber\\
 \text{s.t.}&  \sum_{i\in [n]}x_i = \frac{1+c}{2}\cdot n,&
\quad\quad \forall i\in[n]~x_i\in[0,1]
\end{align}
Again, we let $z\eqdef \delta^{-2/3}/(\delta^{-2/3}+1)<1/2$ for $q=2$ and by Lemma~\ref{optimal value of x_i} and the same argument for the median mechanism, the optimal value of $x_i$ lies in $\{ a,1\}$, where $a\in[0,z]$ and the space of optimal solution can be narrowed down to 2 parameters $a$ and $|\{i: x_i=a\}|$.

Using the constraint $\sum_{i=1}^n x_i = (1+c)\cdot n/2$, the number $|\{i: x_i=a\}|$ and $|\{i: x_i=1\}|$ must be, respectively, $\frac{1-c}{2(1-a)}\cdot n$ and $\frac{1-2a+c}{2(1-a)}\cdot n$. The number $|\{i: x_i=1\}|\geq 0$ implies that $1-2a+c\geq 0$. Therefore, we need to find $\lambda$ such that
\begin{multline}
    \label{eq: robustness optimization over a}   \min_{a\in\left[0,\min\{z,\frac{1+c}{2}\}\right]}\left[\frac{1-c}{2(1-a)}\cdot n\cdot h(a) + \frac{1-2a+c}{2(1-a)}\cdot n\cdot h(1)\right]\\
    =\min_{a\in\left[0,\min\{z,\frac{1+c}{2}\}\right]}\left[\frac{n\cdot\lambda}{2(1-a)}\cdot\left((1-c)\cdot\left(\delta\cdot(1-a)^{1/q}-a^{1/q}\right)-1+2\cdot a-c\right)\right]\ge 0.
\end{multline}
Thus, we need to minimize $$u_2(a)\eqdef(1-c)\cdot\left(\delta\cdot(1-a)^{1/q}-a^{1/q}\right)-1+2\cdot a-c.$$
Let $a_2\eqdef \argmin\limits_{a\in\left[0,\min\{z,\frac{1+c}{2}\}\right]} u_2(a)$. We observe that
\begin{lemma}
    \label{lm:robustness optimal solution for a_2}
    If $u_2(a_2)=0$, then $a_2$ is optimal $\Rightarrow u_2'(a_2)=0\ or\ a_2=\frac{1+c}{2}$.
\end{lemma}
The proof of Lemma~\ref{lm:robustness optimal solution for a_2} is essentially the same as the proof of Lemma~\ref{lm:consistency optimal solution for a_1} (we simply replace $c$ with $-c$ in the proof of Lemma~\ref{lm:consistency optimal solution for a_1}).

Lastly, we need to find $\lambda_2$ (analog of $\lambda^*$ from the proof of Theorem~\ref{thm:UB}) such that the minimum of \eqref{eq:robustness relaxed main} and respectively \eqref{eq: robustness optimization over a} is equal to $0$, which ensures us an $1/\lambda_2$-robustness guarantee. By Lemma~\ref{lm:robustness optimal solution for a_2}, the optimal $\lambda_2$ and $a_2$ must satisfy (i) $u_2(a_2)=0$ and (ii) $u_2'(a_2)=0$ or $a_2=(1+c)/2$. Then by denoting $\delta_2=\delta(\lambda_2)$, we can solve for $a_2$, $\delta_2$, $\lambda_2$ in the same way as before.

First, we compute for $a'_2$ such that $u_2'(a'_2)=0$ and $u_2(a'_2)=0$. The corresponding $a'_2$ is the optimal solution (i.e. $a_2=a'_2$) if $a'_2\leq \frac{1+c}{2}$, or otherwise, $a_2=\frac{1+c}{2}$. The equations for $u_2'(a'_2)=0$ and $u_2(a'_2)=0$ are
\begin{equation}
\label{eq:equations for a_2}
\left\{
    \begin{aligned}
    &\left(1-c\right)\cdot\left(\delta_2\cdot(1-a'_2)^{1/2}-{(a'_2)}^{1/2}\right)-1+2\cdot a'_2-c=0\\
    &\frac{1}{2}\cdot\left(1-c\right)\cdot\left(\delta_2\cdot(1-a'_2)^{-1/2}-(a'_2)^{-1/2}\right)+2=0.
    \end{aligned}
    \right.
\end{equation}
We multiply $2(a'_2-1)$ on both side of the second equation in~\eqref{eq:equations for a_2} and get
\begin{equation}
\label{eq:second equation for a_2}
    \left(1-c\right)\cdot\left(\delta_2\cdot(1-a'_2)^{1/2}-(a'_2)^{1/2}+(a'_2)^{-1/2}\right)+4(a'_2-1)=0.
\end{equation}
The first equation gives us $(1-c)\cdot(\delta_2\cdot(1-a'_2)^{1/2}-(a'_2)^{1/2})=1-2\cdot a'_2+c$. By plugging it in equation~\eqref{eq:second equation for a_2}, we get
\begin{equation}
\label{eq:third equation for a_2}
    2\cdot a'_2+(1-c)\cdot (a'_2)^{-1/2}-3+c=0
\end{equation}
Let $t_2=(a'_2)^{1/2}$, we get
\begin{align*}
&2t_2^2+(1-c)/t_2-3+c=0\\
\Leftrightarrow& 2t_2^3-(3-c)\cdot t_2 +(1-c)=0\\
\Leftrightarrow& (t_2-1)\cdot\left(2t_2^2+2t_2-(1-c)\right)=0\\
\Leftrightarrow& t_2 = \frac{-1+\sqrt{3-2c}}{2}.
\end{align*}
Therefore, $a'_2=t_2^2=(2-c-\sqrt{3-2c})/2$. Since $a'_2\leq\frac{1+c}{2}$ on $c\in[0,1)$, we get $a_2=a_2'$.

Next, by plugging the value of $a_2$ into $u_2(a_2)=0$, we can obtain the value of $\delta_2$ and $\lambda_2$. 

\begin{align*}
  \delta_2^2&=\left(\frac{1-2a_2+c}{1-c}+a_2^{1/2}\right)^2/(1-a_2)\\
  &=\frac{1}{(1-c)^2}\cdot\left(1-2a_2+c+a_2^{1/2}\cdot(1-c)\right)^2/(1-a_2) \\
  &=\frac{1}{(1-c)^2}\cdot\left(\frac{-3+5c+(3-c)\cdot\sqrt{3-2c}}{2}\right)^2/\left(\frac{c+\sqrt{3-2c}}{2}\right) \\
  &=\frac{-2c^3+40c^2-66c+36-(18-36c+10c^2)\cdot\sqrt{3-2c}}{2c\cdot(1-c)^2+2(1-c)^2\cdot\sqrt{3-2c}} \\
  &=\frac{-c^3+20c^2-33c+18-(9-18c+5c^2)\cdot\sqrt{3-2c}}{c\cdot(1-c)^2+(1-c)^2\cdot\sqrt{3-2c}}.
\end{align*}
Therefore,
\begin{align*}
  \lambda_2&=(1+\delta_2^2)^{-1/2}\\
  &=\left(1+\frac{-c^3+20c^2-33c+18-(9-18c+5c^2)\cdot\sqrt{3-2c}}{c\cdot(1-c)^2+(1-c)^2\cdot\sqrt{3-2c}}\right)^{-1/2}\\
  &=\left(\frac{18c^2-32c+18-(4c^2-16c+8)\cdot\sqrt{3-2c}}{c\cdot(c-1)^2+(1-c)^2\cdot\sqrt{3-2c}}\right)^{-1/2}\\
  &=(1-c)\cdot\left(-4\sqrt{3-2c}\cdot c+6\sqrt{3-2c}+10c-8\right)^{-1/2}.
\end{align*}
\end{proof}

\end{document}